\documentclass{amsart}

\usepackage{fancyhdr}
\usepackage{latexsym}
\usepackage{amssymb}
\newtheorem{lemma}{Lemma}[section]
\newtheorem{proposition}[lemma]{Proposition}
\newtheorem{corollary}[lemma]{Corollary}
\newtheorem{theorem}[lemma]{Theorem}
\newtheorem{assumption}[lemma]{Standing Assumption}
\theoremstyle{definition}
\newtheorem{definition}[lemma]{Definition}
\theoremstyle{definition}
\newtheorem{example}[lemma]{Example}
\newtheorem{remark}[lemma]{Remark}
\theoremstyle{definition}

\newcommand\GEX{{\Gamma\!\sb{\rm ex}}}

\title {The exocenter and type decomposition of a generalized pseudoeffect algebra}
\author{David J. Foulis, Sylvia Pulmannov\'a and Elena Vincekov\'a}
\thanks{The second and third authors were supported by  by ERDF OP R \& D metaQUTE
ITMS 26240120022, grant VEGA 2/0059/12 and by Science and Technology Assistance
Agency under the contract no. APVV-0178-11}

\begin{document}

\address{Department of Mathematics an Statistics, Univ. of Massachusetts, Amherst, MA, USA;
\v Stef\'anikova 49, 814 73 Bratislava, Slovakia}
\email{foulis@math.umass.edu; pulmann@mat.savba.sk, vincek@mat.savba.sk}
\maketitle
\markboth{Foulis, D., Pulmannov\'a, S., Vincekov\'a, E.}{The exocenter and type
decompositions for GPEAs }
\date{}

\begin{abstract}  We extend to a generalized pseudoeffect algebra (GPEA) the
notion of the exocenter of a generalized effect algebra (GEA) and show that
elements of the exocenter are in one-to-one correspondence with direct
decompositions of the GPEA; thus the exocenter is a generalization of the center
of a pseudoeffect algebra (PEA). The exocenter forms a boolean algebra and the
central elements of the GPEA correspond to elements of a sublattice of the exocenter
which forms a generalized boolean algebra. We extend to GPEAs the notion of central
orthocompleteness, prove that the exocenter of a centrally orthocomplete GPEA (COGPEA)
is a complete boolean algebra and show that the sublattice corresponding to the center
is a complete boolean subalgebra. We also show that in a COGPEA, every element admits
an exocentral cover and that the family of all exocentral covers, the so-called exocentral
cover system, has the properties of a hull system on a generalized effect algebra. We extend
the notion of type determining (TD) sets, originally introduced for effect algebras and
then extended to GEAs and PEAs, to GPEAs, and prove a type-decomposition theorem, analogous
to the type decomposition of von Neumann algebras.

\end{abstract}

\maketitle

\section{Introduction} \label{sc:Intro} 

\noindent Our purpose in this article is to define and study extensions to generalized pseudoeffect algebras of the notions of the center, central orthocompleteness, central
cover, type determining sets and type decompositions for an effect algebra, resp. for
a pseudoeffect algebra (see \cite{FPType, COEA, ExoCen, CenGEA, TDPA, GFP}).

Effect algebras (EAs) \cite{FandB} were originally introduced as a basis for the
representation of quantum measurements \cite{BLM}, especially  those that involve
fuzziness or unsharpness. Special kinds of effect algebras include orthoalgebras,
MV-algebras, Heyting MV-algebras, orthomodular posets, orthomodular lattices, and
boolean algebras. An account of the axiomatic approach to quantum mechanics employing
EAs can be found in \cite{DvPuTrends}.

Several authors have studied or employed algebraic structures that, roughly speaking,
are EAs ``without a largest element." These studies go back to M.H. Stone's work
\cite{Stone} on generalized boolean algebras; later M.F. Janowitz \cite{Jan}
extended Stone's work to generalized orthomodular lattices. More recent
developments along these lines include \cite{FandB, HedPu, KR, KCh, MI, PV07,
Zdenka99, W}.

The notion of a (possibly) non-commutative effect algebra, called a pseudoeffect
algebra, was introduced and studied in  \cite{DV1, DV2, D}. Whereas a prototypic example
of an effect algebra is the order interval from $0$ to a positive element in a partially
ordered abelian group, an analogous interval in a partially ordered non-commutative group
is a prototype of a pseudoeffect algebra. Pseudoeffect algebras ``without a largest element",
called generalized pseudoeffect algebras, also have been studied in the literature
\cite{DvVepo, DvVegen, PVext, XieLi}.

The classic decomposition of a von Neumann algebra as a direct sum of subalgebras of types
I, II and III  \cite{MvN}, which plays an important role in the theory of von Neumann
algebras, is reflected by a direct sum decomposition of the complete orthomodular lattice
(OML) of its projections. The type-decomposition for a von Neumann algebra is dependent
on the von Neumann-Murray dimension theory, and likewise the early type-decomposition theorems
for OMLs were based on the dimension theories of L. Loomis  \cite{L} and of S. Maeda \cite{M}.
Decompositions of complete OMLs into direct summands with various special properties were
obtained in \cite{CChM, K, R} without explicitly employing lattice dimension theory.
More recent and considerably more general results on type-decompositions based on dimension
theory can be found in \cite{GW}. Dimension theory for effect algebras was developed in
\cite{HandD}.

As a continuation of the aforementioned work, the theory of so called type determining
sets was introduced and applied, first to obtain direct decompositions for centrally
orthocomplete effect algebras \cite{FPType, COEA}, and later for centrally orthocomplete pseudoeffect algebras \cite{TDPA}. While direct decompositions of effect algebras and pseudoeffect algebras are completely described by their central elements \cite{D, GFP},
for the generalized structures without a top element, we need to replace the center by
the so called exocenter, which is composed of special endomorphisms, resp. ideals \cite{ExoCen, Je00}.

The present paper is organized as follows. In Section \ref{sc:GPEAs}, we introduce
basic definitions and facts concerning generalized pseudoeffect algebras (GPEAs). In
Section \ref{sc:ExoCenter} we introduce the notion of the exocenter of a GPEA and
study its properties. Section \ref{sc:CenterGPEA} is devoted to central elements in
a GPEA and relations between the center and the exocenter. The notion of central orthocompleteness is extended to GPEAs in Section \ref{sc:CO} where it is shown that
the center of a centrally orthocomplete GPEA (COGPEA) is a complete boolean algebra.
In Section \ref{sc:ExoCenCover} we introduce the exocentral cover, which extends the
notion of a central cover for an EA. In Section \ref{sc:TDsets}, we develop the theory of type determining sets for GPEAs and show some examples. Finally, in Section
\ref{sc:TypeDecomp}, we develop the theory of type decompositions of COGPEAs into
direct summands of various types. We note that COGPEAs are, up to now, the most
general algebraic structures for which the theory of type determining sets has
been applied to obtain direct decompositions.

\section{Generalized pseudoeffect algebras} \label{sc:GPEAs} 

\noindent We abbreviate `if and only if' as `iff' and the notation $:=$ means
`equals by definition'.

\begin{definition}\label{def:gpea}
A \emph{generalized pseudoeffect algebra} (GPEA) is a partial algebraic structure
$(E,\oplus,0)$, where $\oplus$ is a partial binary operation on $E$ called the
\emph{orthosummation}, $0$ is a constant in $E$ called the \emph{zero element},
and the following conditions hold for all $a,b,c\in E$:
\begin{enumerate}
\item[ ]
\begin{enumerate}
 \item[(GPEA1)] (\emph{associativity}) $(a\oplus b)$ and $(a\oplus b)\oplus c$
 exist iff $b\oplus c$ and $a\oplus (b\oplus c)$ exist and in this case
 $(a\oplus b)\oplus c=a\oplus (b\oplus c)$.
 \item[(GPEA2)] (\emph{conjugacy}) If $a\oplus b$ exists, then there are
 elements $d,e\in E$ such that $a\oplus b=d\oplus a=b\oplus e$.
 \item[(GPEA3)] (\emph{cancellation}) If $a\oplus b=a\oplus c$, or $b
 \oplus a=c\oplus a$, then $b=c$.
 \item[(GPEA4)] (\emph{positivity}) If $a\oplus b=0$, then $a=b=0$.
 \item[(GPEA5)] (\emph{zero element}) $a\oplus 0$ and $0\oplus a$ always
 exist and are both equal to $a$.
\end{enumerate}
   \end{enumerate}
\end{definition}

As a consequence of (GPEA3), the elements $d$ and $e$ in (GPEA2) are
uniquely determined by $a$ and $b$. Following the usual convention,
we often refer to a GPEA $(E,\oplus,0)$ simply as $E$.

If $E$ and $F$ are GPEAs, then a mapping $\phi\colon E\to F$ is a
\emph{GPEA-morphism} iff, for all $a,b\in E$, if $a\oplus b$ exists
in $E$, then $\phi(a)\oplus\phi(b)$ exists in $F$ and $\phi(a\oplus b)
=\phi(a)\oplus \phi(b)$. If $\phi\colon E\to F$ is a bijective
GPEA-morphism and $\phi\sp{-1}\colon F\to E$ is also a GPEA-morphism,
then $\phi$ is a \emph{GPEA-isomorphism}.

\begin{assumption}
In what follows, $(E,\oplus,0)$ is a generalized pseudoeffect algebra.
In general, lower case Latin letters $a,b,c,...,x,y,z$, with or without
subscripts, will denote elements of $E$. If we write an equation
involving an orthosum, e.g. $x\oplus y=z$, we tacitly assume its existence.
\end{assumption}

\begin{definition} \label{df:leqetc}
The relation $\leq$ is defined on the GPEA $E$ by
\[
a\leq b \mbox{ iff } a\oplus x=b \mbox{ for some } x\in E
\]
or equivalently (in view of (GPEA2)), by
\[
a\leq b \mbox{ iff } y\oplus a=b \mbox{ for some } y\in E.
\]
If $a\leq b$, then by (GPEA3) the elements $x$ and $y$ such
that $a\oplus x=y\oplus a=b$ are uniquely determined by $a$
and $b$, and we define the (left and right) differences
\[
a/b:=x \mbox{ and } b\backslash a:=y.
\]
In the event that $a\leq b$ and $a/b$ coincides with $b\backslash a$,
we also define
\[
b\ominus a:=a/b=b\backslash a.
\]
We say that elements $p$ and $q$ in $E$ are \emph{orthogonal}, in
symbols $p\perp q$, iff $p\oplus q$ and $q\oplus p$ both exist and
are equal. The GPEA $E$ is \emph{commutative} iff $p\perp q$ holds
whenever $p\oplus q$ is defined.
\end{definition}

Evidently, if either $a/b$ or $b\backslash a$ exists, then both
exist and $a\leq b$; conversely, if $a\leq b$, then both $a/b$ and
$b\backslash a$ exist and $b=a\oplus(a/b)=(b\backslash a)\oplus a$.
Also, if $b\ominus a$ exists, then $a,b\ominus a\leq b$,
$a\perp(b\ominus a)$ and $a\oplus(b\ominus a)=(b\ominus a)\oplus a
=b$. We note that a commutative GPEA is the same thing as a
\emph{generalized effect algebra} \cite{Zdenka99}.

The GPEA $E$ is partially ordered by $\leq$ and $0$ is the smallest
element in $E$. The cancellation laws in (GPA3) are easily extended
to $\leq$ as follows:
\[
\mbox{If }a\oplus b\leq a\oplus c,\mbox{or if }b\oplus a\leq c\oplus a,
 \mbox{ then }b\leq c.
\]
An existing supremum (resp. infimum) in the partially ordered set (poset)
$E$ of elements $a$ and $b$ is denoted by $a\vee b$ (resp. by $a\wedge b$).
We say that $a$ and $b$ are \emph{disjoint} iff $a\wedge b=0$. We note that
a GPEA-morphism preserves inequalities and corresponding left and right
differences.

An important example of a GPEA (\cite{DvVepo}, Example 2.3) is a subset of
the positive cone in a partially ordered group (po-group). Let $(G,+,0,
\leq)$ be a po-group with $G^+:=\{g\in G: 0\leq g\}$. Let $G^0$ be a
nonempty subset of $G^+$ such that for all $a,b\in G^0$, if $b\leq a$
then $-a+b,\ b-a\in G^0$. Then $(G^0,\oplus,0)$, where $\oplus$ is the
group addition restricted to those pairs of elements whose sum is
again in $G^0$, is a GPEA whose partial order coincides with the group
partial order restricted to $G^0$.

\begin{lemma} \label{le:SlashProps}
Let $a,b,c,d\in E$ with $a\leq b$. Then{\rm:}
\begin{enumerate}
\item $b\backslash a,\  a/b\leq b$ and $(b\backslash a)/b=b\backslash(a/b)=a$.
\item $d\leq a/b\Leftrightarrow a\oplus d\leq b\Leftrightarrow d\leq b$ and
 $a\leq b\backslash d$.
\item If $b\oplus d$ exists, then $a/(b\oplus d)=(a/b)\oplus d$,
 $a\oplus d$ exists, and $a\oplus d\leq b\oplus d$. Also, if $d\oplus b$
 exists, then $(d\oplus b)\backslash a=d\oplus(b\backslash a)$, $d\oplus a$
 exists, and $d\oplus a\leq b\oplus a$.
\item If $a\leq b\leq c$, then $a/c=a/b\oplus b/c$ and $c\backslash a=c
 \backslash b\oplus b\backslash a$.
\end{enumerate}
\end{lemma}

\begin{proof}
(i) As $b=b\backslash a\oplus a$, we get $(b\backslash a)/b=a$, and $b=a
\oplus a/b$ implies $b\backslash (a/b)=a$.

(ii) If $d\leq a/b$, then $\exists x\in E$ with $d\oplus x=a/b$, so $(a\oplus d)
\oplus x=a\oplus(d\oplus x)=b$, and therefore $a\oplus d\leq b$. If $a\oplus d
\leq b$, then $\exists y\in E$, $y\oplus(a\oplus d)=(y\oplus a)\oplus d=b$,
whence $d\leq b$, $y\oplus a=b\backslash d$, and $a\leq b\backslash d$. Thus
$d\leq a/b\Rightarrow a\oplus d\leq b\Rightarrow d\leq b\text{\ and\ }a\leq b
\backslash d$. Proofs of the converse implications are straightforward.

(ii) Assume that $b\oplus d$ exists. Then $a\leq b\leq b\oplus d$ and $a
\oplus a/(b\oplus d)=b\oplus d=(a\oplus a/b)\oplus d=a\oplus((a/b)\oplus d)$,
whence $a/(b\oplus d)=(a/b)\oplus d$ by cancellation. Also, as $a\leq b$, we
have $b\backslash a\oplus a=b$, whence $b\oplus d=(b\backslash a\oplus a)
\oplus d=b\backslash a\oplus(a\oplus d)$, whence $a\oplus d$ exists and
$a\oplus d\leq b\oplus d$. The remaining assertion is proved analogously.

(iv) As $a\leq b\leq c$, we have $a\oplus(a/b\oplus b/c)=(a\oplus a/b)
\oplus b/c=b\oplus b/c=c=a\oplus a/c$, whence $a/b\oplus b/c=a/c$ by
cancellation. The second equality is proved similarly.
\end{proof}

\begin{lemma} \label{le:oplusdist}
Let $e\in E$, and let $(f\sb{i})\sb{i\in I}$ be a family of elements of $E$
such that the supremum $f:=\bigvee\sb{i\in I}f\sb{i}$ exists in $E$. Suppose
that $e\oplus f$ {\rm(}resp. $f\oplus e${\rm)} exists. Then $e\oplus f\sb{i}$
{\rm(}resp. $f\sb{i}\oplus e${\rm)} exists for all $i\in I$, the supremum
$\bigvee\sb{i\in I}(e\oplus f\sb{i})$ {\rm(}resp. the supremum $\bigvee\sb
{i\in I}(f\sb{i}\oplus e)${\rm)} exists in $E$, and $e\oplus f=\bigvee\sb{i
\in I}(e\oplus f\sb{i})$ {\rm(}resp. $f\oplus e=\bigvee\sb{i\in I}(f\sb{i}
\oplus e)${\rm)}.
\end{lemma}

\begin{proof}
We prove the lemma under the hypothesis that $e\oplus f$ exists. The proof
under the alternative hypothesis is similar. For each $i\in I$, we have
$f\sb{i}\leq f$, and therefore $e\oplus f\sb{i}$ exists and $e\oplus f
\sb{i}\leq e\oplus f$ (Lemma \ref{le:SlashProps} (iii)). Suppose that
$e\oplus f\sb{i}\leq b\in E$ for all $i\in I$, i.e., there exists $x\sb{i}$
with $b=(e\oplus f\sb{i})\oplus x\sb{i}=e\oplus (f\sb{i}\oplus x\sb{i})$.
Then $e\leq b$ and $f\sb{i}\leq f\sb{i}\oplus x\sb{i}=e/b$ for all $i\in I$,
whence $f\leq e/b$, and it follows from Lemma \ref{le:SlashProps} (ii) that $e
\oplus f\leq b$, proving that $e\oplus f=\bigvee\sb{i\in I}(e\oplus f\sb{i})$.
\end{proof}

By (GPEA1), we may omit parentheses in expressions such as $a\oplus b\oplus c$.
By recursion, the partial operation $\oplus$ can be extended to finite
sequences $e_1,e_2,\ldots,e_n$ as follows: The orthosum $e_1\oplus\cdots
\oplus e_n$ exists iff the elements $f:=e_1\oplus e_2\oplus\cdots\oplus
e_{n-1}$ and $f\oplus e_n$ both exist, and then $e_1\oplus\cdots\oplus e_n
:=f\oplus e_n$. In general, the orthosum may depend on the order of its
orthosummands.

In a similar way, by recursion, we also define \emph{orthogonality} and
the corresponding \emph{orthosum} for a finite sequence of elements in
$E$, and it turns out that the orthosum does not depend on the order of
the orthosummands. Therefore, in the obvious way, we define orthogonality and
the corresponding orthosum for finite families in $E$. (We understand that
the empty family in $E$ is orthogonal and that its orthosum is $0$.) The
notion of orthogonality and the orthosum for arbitrary families is defined
as follows: A family $(e_i)_{i\in I}$ in $E$ is said to be \emph{orthogonal}
iff every finite subfamily $(e_i)_{i\in F}$ ($I\supseteq F$ is finite) is
orthogonal in $E$. The family $(e_i)_{i\in I}$ is \emph{orthosummable} with
\emph{orthosum} $\oplus_{i\in I} e_i$ iff it is orthogonal and the supremum
$\bigvee_{F\subseteq I}(\oplus_{i\in F} e_i)$ over all finite subsets $F$ of $I$
exists in $E$, in which case $\oplus_{i\in I}e_i:=\bigvee_{F\subseteq I}
(\oplus_{i\in F} e_i)$.

\begin{lemma}\label{le:veeopluswedge}
Let $e,f\in E$. If $e\perp f$ and $e\vee f$ exists in $E$, then $e\wedge f$
exists in $E$, $(e\vee f)\perp(e\wedge f)$, and $e\oplus f=(e\vee f)
\oplus(e\wedge f)$.
\end{lemma}

\begin{proof}
As $e\perp f$, we have $e\oplus f=f\oplus e$. Evidently $e,f\leq e\oplus f$,
so $e\leq e\vee f\leq e\oplus f$, and by Lemma \ref{le:SlashProps} (iv),
$e/(e\vee f)\oplus(e\vee f)/(e\oplus f)=e/(e\oplus f)=f$, whence $(e\vee f)
/(e\oplus f)\leq f$. Likewise, $(e\vee f)/(e\oplus f)\leq e$. Suppose that
$d\leq e,f$. By Lemma \ref{le:SlashProps} (iii), $f\leq f\oplus(e\backslash d)
=(f\oplus e)\backslash d=(e\oplus f)\backslash d$. Likewise, $e\leq(e\oplus f)
\backslash d$, and we have $e\vee f\leq(e\oplus f)\backslash d$; hence by
Lemma \ref{le:SlashProps} (ii), $d\leq (e\vee f)/(e\oplus f)$. This proves
that $(e\vee f)/(e\oplus f)=e\wedge f$, from which we obtain $e\oplus f=
(e\vee f)\oplus (e\wedge f)$. Similarly, by considering $(e\oplus f)
\backslash(e\vee f)$, which is again under $e$ and $f$, and arguing that
$(e\oplus f)\backslash(e\vee f)=e\wedge f$, we find that $e\oplus f=
(e\wedge f)\oplus (e\vee f)$.
\end{proof}

\begin{definition}\label{def:pea}
A \emph{pseudoeffect algebra} (PEA) is a partial algebraic structure $(E,\oplus,0,1)$,
where $\oplus$ is a partial operation and $0$ and $1$ are constants, and the following hold:
\begin{enumerate}
\item[ ]
\begin{enumerate}
\item[(PEA1)] $a\oplus b$ and $(a\oplus b)\oplus c$ exist iff $b\oplus c$ and $a\oplus
 (b\oplus c)$ exist, and in this case $(a\oplus b)\oplus c=a\oplus (b\oplus c)$.
\item[(PEA2)] There is exactly one $d\in E$ and exactly one $e\in E$ such that $a
 \oplus d=e\oplus a=1$.
\item[(PEA3)] If $a\oplus b$ exists, there are elements $d,e\in E$  such that
 $a\oplus b=d\oplus a=b\oplus e$.
\item[(PEA4)] If $1\oplus a$ or $a\oplus 1$ exists, then $a=0$
\end{enumerate}
\end{enumerate}
\end{definition}

The partial ordering for a PEA is defined in the same way as the
partial ordering for a GPEA. It is easy to see, that a PEA is the
same thing as a GPEA with a greatest element. We claim the following
statement from (\cite{DvVepo}, Proposition 2.7):

\begin{proposition} \label{pr:interval}
Let $(E,\oplus,0)$ be a GPEA and let $u\in E$. Then $(E[0,u],\oplus\sb{u},0,u)$ is
a PEA, where $E[0,u]:=\{ a\in E:a\leq u\}$ and where $a\oplus\sb{u}b$ is defined
for $a,b\in E[0,u]$ iff $a\oplus b$ exists in $E$ and $a\oplus b\leq u$, in which
case $a\oplus\sb{u}b:=a\oplus b$.
\end{proposition}

\begin{definition} \label{df:Ideal}
An \emph{ideal} of the GPEA $E$ is a nonempty subset $I\subseteq E$ such
that:
\begin{enumerate}
\item[(I1)] If $a\in I$, $b\in E$, and $b\leq a$, then $b\in I$.
\item[(I2)] If $a,b\in I$ and $a\oplus b$ exists, then $a\oplus b\in I$.
\end{enumerate}
If $I$ is an ideal in $E$, then $I$ is said to be \emph{normal} iff,
\begin{enumerate}
\item[(N)] whenever $a,x,y\in E$ and $a\oplus x=y\oplus a$, then
 $x\in I\Leftrightarrow y\in I$.
\end{enumerate}
\end{definition}

\begin{definition} \label{df:CentralIdeal}
We say, that an ideal $S$ in the GPEA $E$ is \emph{central}, or equivalently,
that it is a \emph{direct summand} of $E$, iff there is an ideal $S'$ in $E$
such that
\begin{enumerate}
\item[(1)] $a\in S, b\in S'\Rightarrow a\perp b$, and
\item[(2)] every $a\in E$ can be uniquely written a an
 orthosum $a=a_1\oplus a_2$ with ``coordinates" $a_1\in S$
 and $a_2\in S'$.
\end{enumerate}
We write $E=S\oplus S'$ iff (1) and (2) hold.
\end{definition}

If $E=S\oplus S'$, then $S'$ is also a central ideal (direct summand)
in $E$, $S'$ is uniquely determined by $S$ (cf. the proof of \cite
[Lemma 4.3]{CenGEA}), and all GPEA calculations on $E$ can be conducted
``coordinatewise" in the obvious sense. If $E=S\oplus S'$, we refer to
$S$ and $S'$ as \emph{complementary direct summands} of $E$.

\begin{proposition}  \label{pr:DirSumNormal}
Any central ideal {\rm(}direct summand{\rm)} of a GPEA $E$ is normal.
\end{proposition}

\begin{proof}
Let $S$ be a central ideal of $E$ with $S'$ as its complementary direct
summand, and assume that $a,x,y\in E$ with $a\oplus x=y\oplus a$. We can
write $a$ uniquely as $a=a\sb{1}\oplus a\sb{2}$ with $a\sb{1}\in S$ and
$a\sb{2}\in S'$. Then $a\sb{1}\oplus a\sb{2}\oplus x=y\oplus a\sb{1}
\oplus a\sb{2}$. Suppose that $x\in S$. Then, as $a\sb{2}\in S'$, we
have $x\perp a\sb{2}$, so $a\sb{2}\oplus x=x\oplus a\sb{2}$, whence
$a\sb{1}\oplus x\oplus a\sb{2}=y\oplus a\sb{1}\oplus a\sb{2}$, and by
cancellation $a\sb{1}\oplus x=y\oplus a\sb{1}$. Therefore, $y\leq a\sb{1}
\oplus x\in S$, and it follows that $y\in S$. By a similar argument, if
$y\in S$, then $x\in S$.
\end{proof}

The notion that $E$ is a direct sum $E=S\oplus S'$ of two central ideals
is extended to finitely many direct summands $E=S\sb{1}\oplus S\sb{2}
\oplus\cdots\oplus S\sb{n}$ in the obvious way, each $S\sb{i}$, $i=
1,2,...,n$, being a central ideal (direct summand) in $E$ with complementary
direct summand $(S\sb{i})'=S\sb{1}\oplus\cdots S\sb{i-1}\oplus S\sb{i+1}
\cdots\oplus S\sb{n}$.

\section{The exocenter of a GPEA} \label{sc:ExoCenter} 

\begin{definition} \label{df:ExoCen}
The exocenter of the GPEA $E$, denoted by $\GEX(E)$, is the set of all mappings
$\pi: E\rightarrow E$ such that for all $e,f\in E$ the following hold:
\begin{enumerate}
\item[ ]
\begin{enumerate}
\item[(EXC1)] $\pi: E\rightarrow E$ is a PGEA-endomorphism of $E$, that is: if
 $e\oplus f$ exists, then $\pi e\oplus \pi f$ exists and $\pi(e\oplus f)=\pi e
 \oplus\pi f$.
\item[(EXC2)] $\pi$ is idempotent (i.e., $\pi (\pi e)=\pi e$).
\item[(EXC3)] $\pi$ is decreasing (i.e., $\pi e\leq e$).
\item[(EXC4)] $\pi$ satisfies the following orthogonality condition: if $\pi e=e$
 and $\pi f=0$, then $e\perp f$ (i.e., $e\oplus f=f\oplus e$).
\end{enumerate}
\end{enumerate}
If $\pi\in \GEX(E)$ and $e\in E$, then as $\pi e\leq e$ by (EXC3), we can (and
do) define $\pi\,'e:=(\pi e)/e$ for all $e\in E$.
\end{definition}

\begin{lemma} \label{le:piprime}
If $\pi\in\GEX(E)$ and $e\in E$, then $\pi\,'e=(\pi e)/e=e\backslash(\pi e)
=e\ominus\pi e$ and $\pi e\perp\pi\,' e\text{\ with\ }\pi e\oplus\pi\,' e
=\pi\,'e\oplus\pi e=e.$
\end{lemma}

\begin{proof}
Let $\pi\in\GEX(E)$ and $e\in E$. As $\pi e\leq e$, both $\pi\,'e=(\pi e)
/e$ and $e\backslash(\pi e)$ are defined, and with $x:=(\pi e)/e$
and $y:=e\backslash(\pi e)$, we have $\pi e\oplus x=e=y\oplus\pi e$. We
apply the mapping $\pi$ and obtain $\pi e\oplus\pi x=\pi e=\pi y\oplus
\pi e$; hence $\pi x=\pi y=0$ and by (EXC4), $\pi e\oplus x=x\oplus\pi e
=e$ and also $\pi e\oplus y=y\oplus\pi e=e$. Therefore by cancellation,
$\pi\,'e=(\pi e)/e=x=y=e\backslash(\pi e)=e\ominus\pi e$, and $\pi e
\perp\pi\,'e$ with $\pi e\oplus\pi\,'e=\pi\,'e\oplus\pi e=e.$
\end{proof}

\begin{theorem}\label{th:EXCprop}
If $\pi\in \GEX(E)$, then for all $e,f\in E$ the following hold{\rm:}
\begin{enumerate}
\item $\pi(\pi\,'e)=\pi\,'(\pi e)=0$.
\item $\pi\,'\in \GEX(E)$ and $(\pi\,')'=\pi$.
\item If $e\leq \pi f$, then $e=\pi e$.
\item If $e\leq f$, then $\pi e=e\wedge\pi f$.
\item $\pi(E):=\{ \pi e: e\in E\} =\{ e\in E: e=\pi e\}$ is an ideal
 in $E$.
\item $\pi(E)$ is sup/inf-closed in $E$ {\rm(}i.e., $\pi(E)$ is closed under
 the formation of existing suprema and infima in $E$ of nonempty families in
 $\pi(E)${\rm)}.
\item If $e\in \pi (E)$ and $f\in \pi\,'(E)$, then $e\perp f, e\oplus f=
 e\vee f$ and $e\wedge f=0$.
\item For each element $e\in E$ there are uniquely determined elements
 $e_1\in \pi (E), e_2\in \pi\,'(E)$ such that $e=e\sb{1}\oplus e\sb{2}$; in fact,
 $e\sb{1}=\pi e$ and $e\sb{2}=\pi\,'e$.
\item If $e=e_1\oplus e_2, f=f_1\oplus f_2$, where $e_1,f_1\in \pi (E)$,
 $e_2,f_2\in \pi\,'(E)$, then $e\oplus f$ exists iff both $e_1\oplus f_1$ and
 $e_2\oplus f_2$ exist.
\item $\pi\,'(E)=\{ f\in E: f\wedge e=0,\,\forall e\in \pi (E)\}$.
\end{enumerate}
\end{theorem}

\begin{proof}
(i) $\pi(\pi\,'e)=\pi(e\backslash\pi e)=\pi e\backslash\pi\pi e=\pi e\backslash
\pi e=0$ and $\pi\,'(\pi e)=\pi e\backslash\pi\pi e=0$ too.

(ii) By Lemma \ref{le:SlashProps} (i), $(\pi\,')' e=e\backslash\pi\,'e=e
\backslash(\pi e/e)=\pi e$. To prove that $\pi\,'$ is a GPEA-endomorphism
of $E$, suppose that $e\oplus f$ exists. Then by (EXC1) $\pi\,'(e\oplus f)
=(e\oplus f)\backslash(\pi e\oplus \pi f)$, whence $\pi\,'(e\oplus f)
\oplus\pi e\oplus\pi f=e\oplus f$ and so by Lemma \ref{le:SlashProps} (iii),
$\pi\,'(e\oplus f)\oplus\pi e=(e\oplus f)\backslash\pi f=e\oplus(f\backslash
\pi f)=e\oplus\pi\,'f$. As $\pi\pi e =\pi e$ and by (i), $\pi\pi\,'(e\oplus f)
=0$, we have $e\perp\pi\,'f$ by (EXC4), whence $\pi e\oplus\pi\,'(e\oplus f)=
\pi\,'(e\oplus f)\oplus\pi e=e\oplus\pi\,'f$, i.e., $\pi\,'(e\oplus f)=
\pi e/(e\oplus\pi\,'f)$, and a second application of Lemma \ref{le:SlashProps}
(iii) yields $\pi\,'(e\oplus f)=(\pi e/e)\oplus\pi\,'f=\pi\,'e\oplus\pi\,'f$.
Thus, $\pi\,'$ satisfies (EXC1). Moreover, by (i), $\pi\,'(\pi\,'e)=\pi\,'
(e\backslash\pi e)=\pi\,'e\backslash\pi\,'\pi e=\pi\,'e$, whence $\pi\,'$
satisfies (EXC2). Obviously, (EXC3) holds for $\pi\,'$. Finally to prove
that $\pi\,'$ satisfies (EXC4), suppose that $\pi\,'e=e$ and $\pi\,'f
=0$. Then $\pi e=0$ because $\pi\,'e=e\backslash\pi e=e$ and $\pi f=f$ because
$\pi\,'f=f\backslash\pi f=0$. Therefore, since $\pi$ satisfies (EXC4), we have
$e\perp f$, and $\pi\,'$ also satisfies (EXC4). Therefore, $\pi\,'\in\GEX(E)$.

(iii) If $e\leq \pi f$, then $e\backslash\pi e=\pi\,'e\leq \pi\,'\pi f=0$,
whence $e=\pi e$.

(iv) Suppose that $e\leq f$. Then $\pi e\leq e$ and $\pi e\leq \pi f$. Suppose
that $d\leq e,\pi f$. Since $d\leq\pi f$, (iii) implies that $d=\pi d\leq
\pi e$, so $\pi e=e\wedge\pi f$.

(v) If $e=\pi e$, then $e\in\pi(E)$. Vice versa, if $e\in\pi(E)$, then $e=
\pi f$ for some $f\in E$, so $\pi e=\pi\pi f=\pi f=e$, and we have $\pi(E)=
\{e\in E: e=\pi e\}$.

(vi) Assume that $(e_i)_{i\in I}\subseteq\pi(E)$ and $e=\bigvee_{i\in I} e_i$
exists in $E$. As $e_i\leq e$, we have $e_i=\pi e_i\leq\pi e$ for all $i\in I$,
whence $e\leq\pi e$. But also $\pi e\leq e$ and thus $\pi e=e\in\pi(E)$. Since
$\pi(E)$ is an ideal, it is automatically closed under the formation of
existing infima in $E$ of nonempty families in $\pi(E)$.

(vii) Let $e\in\pi(E)$ and $f\in\pi\,'(E)$. Then $e=\pi e$, and $\pi f=\pi
\pi\,'f=0$, whence by (EXC4) $e\perp f$. Clearly, $e,f\leq e\oplus f$. If now
$e,f\leq d\in E$, then $e=\pi e\leq\pi d$ and $f=\pi\,'f\leq\pi\,'d$, thus
$e\oplus f\leq\pi d\oplus\pi\,'d=d$, whence $e\oplus f=e\vee f$. Finally,
by Lemma \ref{le:veeopluswedge}, $e\wedge f=0$.

(viii) Obviously, $e=\pi e\oplus\pi\,'e$, $\pi e\in\pi(E)$ and $\pi\,'e\in
\pi\,'(E)$. Suppose $e=e_1\oplus e_2$ with $e_1\in\pi(E)$, $e_2\in\pi\,'(E)$.
Then $e_1=\pi e_1$, $e_2=\pi\,'e_2$, $\pi e=\pi e_1\oplus\pi e_2=e_1$, and
$\pi\,'e=\pi\,'e_1\oplus\pi\,'e_2=e_2$.

(ix) Suppose $e=e_1\oplus e_2$ and $f=f_1\oplus f_2$, with $e_1,f_1\in\pi(E)$,
$e_2,f_2\in\pi\,'(E)$. If $e_1\oplus f_1$ and $e_2\oplus f_2$ both exist, then $e_1
\oplus f_1\in\pi(E)$ and $e_2\oplus f_2\in\pi\,'(E)$ by (v). Then by (vii) $(e_1
\oplus f_1)\perp(e_2\oplus f_2)$, so $(e_1\oplus f_1)\oplus(e_2\oplus f_2)$
exists and equals $e_1\oplus(f_1\oplus e_2)\oplus f_2=e_1\oplus(e_2\oplus f_1)
\oplus f_2=(e_1\oplus e_2)\oplus(f_1\oplus f_2)=e\oplus f$. If, on the other hand,
$e\oplus f$ exists, then $e_1\oplus e_2\oplus f_1\oplus f_2$ exists and equals $e_1
\oplus f_1\oplus e_2\oplus f_2$, which implies that $e_1\oplus f_1$ and $e_2\oplus f_2$
both exist.

(x) Assume that $f\wedge e=0$ for all $e\in\pi(E)$. As $f=f_1\oplus f_2$ with $f_1\in
\pi(E)$, $f_2\in\pi\,'(E)$, we have $f_1=f\wedge f_1=0$, whence $f=f_2\in\pi\,'(E)$. The
converse follows from (vii).
\end{proof}

\begin{lemma}\label{le:circ}
Let $\xi,\pi\in\GEX (E)$. Then{\rm:}
\begin{itemize}
\item[(i)] $\xi\circ\pi=\pi\circ\xi\in\GEX (E)$.
\item[(ii)] $\xi=\xi\circ\pi\,\Leftrightarrow\,\xi e\leq\pi e,\,\forall\, e\in E\,
 \Leftrightarrow\,\xi (E)\subseteq\pi (E)$.
\end{itemize}
\end{lemma}

\begin{proof}
(i) Since $\xi(\pi e)\leq\pi e$, part (iii) of Theorem \ref{th:EXCprop} yields
$\xi(\pi e)=\pi(\xi(\pi e))$. Also, since $\pi e\leq e$ and both, $\pi$ and
$\xi$ are order-preserving mappings, it follows that $\xi(\pi e)=\pi(\xi(\pi e))
\leq\pi(\xi e)$. By symmetry $\pi(\xi e)\leq\xi(\pi e)$, which gives $\xi\circ
\pi=\pi\circ\xi$.

Obviously, $\xi\circ\pi$ is a GPEA-endomorphism. Furthermore, $(\xi\circ\pi)
\circ(\xi\circ\pi)=\xi\circ\pi\circ\pi\circ\xi=\xi\circ\pi\circ\xi=\xi\circ
\xi\circ\pi=\xi\circ\pi$, whence $\xi\circ\pi$ is idempotent. Moreover,
$(\xi\circ\pi)e=\xi(\pi e)\leq\pi e\leq e$, so (EXC3) holds. Finally, suppose
that $e,f\in E$ with $e=\xi(\pi e)$ and $\xi(\pi f)=0$. Then $e=\pi(\xi e)$, so
$e=\xi e=\pi e$. We put $d:=\pi f$, so that $\xi d=0$, $d\leq f$, $d=\pi d=
\pi f$, and $\pi\,'f=(\pi f)/f=f\backslash\pi f=d/f=f\backslash d$. Therefore,
$\pi\,'f\oplus d=d\oplus\pi\,'f=f$. As $e=\xi e$ and $\xi d=0$, (EXC4) implies
that $e\perp d$, i.e., $e\oplus d=d\oplus e$. Also, $\pi(e\oplus d)=\pi e
\oplus\pi d=e\oplus d$, and it follows from $\pi(\pi\,'f)=0$ and (EXC4) that
$(e\oplus d)\perp\pi\,'f$. Consequently,
\[
e\oplus f=e\oplus d\oplus\pi\,'f=\pi\,'f\oplus e\oplus d=\pi\,'f\oplus d
\oplus e=f\oplus e,
\]
proving that $\xi\circ\pi$ satisfies (EXC4).

(ii) If $\xi=\xi\circ\pi$, then $\xi e=\xi(\pi e)\leq\pi e$ for all $e\in E$.
Conversely, if $\xi e\leq\pi e$ for all $e\in E$, then $\xi e=\xi(\xi e)\leq
\xi(\pi e)$. Also, as $\xi(\pi e)\leq\xi e$ always holds, $\xi e=\xi(\pi e)$
for all $e\in E$, which means that $\xi=\xi\circ\pi$. Now if $\xi e\leq\pi e$
for all $e\in E$, then if $e\in\xi(E)$, we get $e=\xi e\leq\pi e$, whence
$\pi e=e\in\pi(E)$. Conversely, if $\xi (E)\subseteq\pi(E)$, then every $\xi e
\in\pi(E)$, thus by (i), $\xi e=\pi(\xi e)=\xi(\pi e)\leq\pi e$.
\end{proof}

\begin{theorem}\label{th:boolalg}
Let $\pi,\xi\in\GEX (E)$ and let $e\in E$. Then $\GEX(E)$ is partially ordered
by $\xi\leq\pi\Leftrightarrow\,\xi=\xi\circ\pi\,\Leftrightarrow\,\xi e\leq\pi e,
\,\forall e\in E\,\Leftrightarrow\,\xi (E)\subseteq\pi (E)$, with $0$ {\rm(}the
zero mapping{\rm)} as the smallest element and $1$ {\rm(}the identity
mapping{\rm)} as the largest element. Moreover, $\GEX(E)$  is a boolean algebra
with $\pi\mapsto\pi\,'$ as the boolean complementation, with $\pi\wedge\xi=
\pi\circ\xi=\xi\circ\pi$, and with $\pi\vee\xi=(\pi\,'\circ\xi\,')'$.
\end{theorem}

\begin{proof}
Let $\pi,\xi\in\GEX (E)$. By Lemma \ref{le:circ}, $\leq$ is a partial order on
$\GEX(E)$ and $0\leq\pi\leq 1$ holds for every $\pi\in\GEX(E)$. Clearly, $\pi
\circ\xi$ is the infimum $\pi\wedge\xi$ of $\pi$ and $\xi$ in $\GEX (E)$. We
also have $\pi\wedge\xi=0$ iff $\pi(\xi e)=0$ for every $e\in E$, which is
equivalent to $\pi(e\backslash\xi e)=\pi e,\,\forall\, e\in E$. But this means
that $\pi(\xi\,'e)=\pi e,\,\forall\, e\in E$, that is $\pi\circ\xi\,'=\pi$, which
holds iff $\pi\leq\xi\,'$. So by \cite[Theorem 4, p. 49]{GrGLT}, $\GEX (E)$ is a
boolean algebra, $\pi\,'$ is the complement of $\pi$ in $\GEX(E)$, and $\pi
\vee\xi=(\pi\,'\circ\xi\,')'$.
\end{proof}

\begin{lemma}\label{le:DisjointPiXi}
Let $\pi,\xi\in\GEX (E)$ with $\pi\wedge\xi=0$ and let $e,f\in E$. Then{\rm:}
\begin{itemize}
\item[(i)] If $e\in\pi(E), f\in\xi(E)$, then $e\perp f$ and $e\oplus f\in
 (\pi\vee\xi)(E)$, $e\oplus f=e\vee f$ and $e\wedge f=0$.
\item[(ii)] $\pi e\perp\xi e, (\pi\vee\xi)e=\pi e\vee\xi e=\pi e\oplus\xi e$
 and $\pi e\wedge\xi e=0$.
\end{itemize}
\end{lemma}

\begin{proof}
(i) By the hypotheses $e=\pi e$ and $f=\xi f$. As $\pi f=\pi(\xi f)=0$ (by
Theorem \ref{th:boolalg}), we get $\pi\,'f=f\backslash\pi f=f$. Therefore,
$f\in\pi\,'(E)$, and by Theorem \ref{th:EXCprop} (vii), $e\perp f$, $e\oplus f
=e\vee f$ and $e\wedge f=0$. Also $e=\pi e\leq(\pi\vee\xi)e\leq e$, whence
$(\pi\vee\xi)e=e$. Likewise, $(\pi\vee\xi)f=f$, whence $e\oplus f=(\pi\vee
\xi)(e\oplus f)\in(\pi\vee\xi)(E)$.

(ii) We need only replace $e$ by $\pi e$ and $f$ by $\xi e$ in (i) to obtain
$\pi e\perp\xi e$, $\pi e\oplus\xi e=\pi e\vee\xi e$ and $\pi e\wedge\xi e=0$.
As $\pi\wedge\xi=0$ in the boolean algebra $\GEX(E)$, we have $\pi\leq\xi\,'$,
whence $\pi e=(\pi\wedge\xi\,')e=(\pi\circ\xi\,')e=\pi(\xi\,'e)$. Thus,
combining the equalities $\xi e\oplus \xi\,'e=e$ and $\pi e\oplus(\pi\,'
\circ\xi\,')e=\pi(\xi\,'e)\oplus\pi\,'(\xi\,'e)=\xi\,'e$, we obtain $\xi e
\oplus\pi e\oplus(\pi\,'\circ\xi\,')e=e$. Therefore, as $(\pi\,'\circ
\xi\,')'e\oplus(\pi\,'\circ\xi\,')e=e$, we infer by cancellation that
$(\pi\vee\xi)e=(\pi\,'\circ\xi\,')'e=\pi e\oplus\xi e=\pi e\vee\xi e$.
\end{proof}

\begin{theorem}\label{th:FinitePwiseDisjointPi}
Let $\pi_1,\pi_2,\ldots ,\pi_n$ be pairwise disjoint elements of the boolean
algebra $\GEX(E)$ and let $e\in E, e_i\in\pi_i(E)$ for $i=1,2,\ldots ,n$.
Then{\rm:}
\begin{itemize}
\item[(i)] $(e_i)_{i=1,2,\ldots ,n}$ is an orthogonal sequence in $E$ and
 $\oplus_{i=1}^n e_i=\bigvee_{i=1}^n e_i$.
\item[(ii)] $(\pi_ie)_{i=1}^n$ is an orthogonal sequence in $E$ and $(\pi_1\vee
 \pi_2\vee\dots\vee\pi_n)e=\oplus_{i=1}^n\pi_i e=\bigvee_{i=1}^n\pi_i e$.
\end{itemize}
\end{theorem}

\begin{proof}
For $n=1$ the assertions hold trivially, and the results for $n=2$ are
consequences  of Lemma \ref{le:DisjointPiXi}. The results for an arbitrary
$n\in\mathbb{N}$ then follow from a straightforward induction argument.
\end{proof}

\begin{theorem} \label{th:finitepointwisesup/inf}
Let $\pi_1,\pi_2,\ldots ,\pi_n\in\GEX(E)$, $e\in E$. Then{\rm:}
\begin{itemize}
\item[(i)] $(\pi_1\wedge\pi_2\wedge\dots\wedge\pi_n)e=\pi_1 e\wedge\pi_2 e
 \wedge\dots\wedge\pi_n e$.
\item[(ii)] $(\pi_1\vee\pi_2\vee\dots\vee\pi_n)e=\pi_1 e\vee\pi_2 e\vee\dots
 \vee\pi_n e$.
\end{itemize}
\end{theorem}

\begin{proof}
We will prove the assertions for $n=2$ and the general cases will then follow
by induction.

(i) Obviously, $(\pi\wedge\xi)e\leq\pi e,\xi e$. Suppose now that $f\leq
\pi e,\xi e$. Then $f=\pi f=\xi f$ by Theorem \ref{th:EXCprop} (iii) and
therefore $f=(\pi\circ\xi)f\leq(\pi\circ\xi)\xi e=(\pi\circ\xi\circ\xi)e=
(\pi\circ\xi)e=(\pi\wedge\xi)e$.

(ii) Working in the boolean algebra $\GEX(E)$, we can write $\pi\vee\xi$ as a
pairwise disjoint supremum:
$$\pi\vee\xi=(\pi\wedge\xi)\vee(\pi\wedge\xi\,')\vee(\pi\,'\wedge\xi).$$
Then we use Theorem \ref{th:FinitePwiseDisjointPi} to get
$$(\pi\vee\xi)e=(\pi\wedge\xi)e\vee(\pi\wedge\xi\,')e\vee(\pi\,'\wedge\xi)e$$
where $\pi e=(\pi\wedge\xi)e\vee(\pi\wedge\xi\,')e$ and $\xi e=(\pi\wedge\xi)e
\vee(\pi\,'\wedge\xi)e$. Therefore $(\pi\vee\xi)e=\pi e\vee\xi e$.
\end{proof}

As is easily confirmed, a cartesian product of GPEAs, with the obvious pointwise
operations and relations, is again a GPEA.

\begin{theorem}\label{th:finitecartesianprod}
Let $\pi_1,\pi_2,\ldots ,\pi_n$ be pairwise disjoint elements of $\GEX(E)$ such
that $\pi_1\vee\pi_2\vee\ldots\vee\pi_n=1$ and let $X$ be the cartesian product
of $\pi_i(E)$ for $i=1,2\ldots ,n$. Then for $(e_1,e_2,\ldots ,e_n)\in X$, the
sequence $(e_i)_{i=1}^n$ is orthogonal in $E$ and $\oplus_{i=1}^n e_i=
\bigvee_{i=1}^n e_i$. Moreover, $\Phi: X\rightarrow E$ defined by $\Phi
(e_1,e_2,\ldots ,e_n):=e_1\oplus e_2\oplus\ldots\oplus e_n$, is a
GPEA-isomorphism and for every $e\in E$, $\Phi^{-1}e=(\pi_1 e,\pi_2 e,\ldots,
\pi_n e)\in X$.
\end{theorem}

\begin{proof}
The first part has already been proved in Theorem
\ref{th:FinitePwiseDisjointPi}. To prove that $\Phi$ is a GPEA-morphism,
let $(e_1,e_2,\ldots ,e_n), (f_1,f_2,\ldots ,f_n)\in X$ and let $e_i
\oplus f_i$ exist for all $i=1,2,\ldots ,n$. Then $(e_1\oplus f_1,e_2
\oplus f_2,\ldots ,e_n\oplus f_n)\in X$ and so $(e_i\oplus f_i)_{i=1}^n$
is an orthogonal sequence. Using Theorem \ref{th:EXCprop} (ix) and
induction, we get $\oplus_{i=1}^n(e_i\oplus f_i)=\oplus_{i=1}^{n-1}(e_i
\oplus f_i)\oplus e_n\oplus f_n=(\oplus_{i=1}^{n-1} e_i)\oplus(\oplus_{i=1}
^{n-1} f_i)\oplus e_n\oplus f_n$. But, since $f_i$ for $i=1,2\ldots ,n-1$ are
all orthogonal to $e_n$, we have $(\oplus_{i=1}^{n-1} e_i)\oplus(\oplus_{i=1}
^{n-1} f_i)\oplus e_n\oplus f_n=(\oplus_{i=1}^{n-1} e_i)\oplus e_n\oplus
(\oplus_{i=1}^{n-1} f_i)\oplus f_n=(\oplus_{i=1}^n e_i)\oplus (\oplus_{i=1}
^n f_i)$, whence $\Phi: X\rightarrow E$ is a GPEA-morphism. Define $\Psi:
E\to X$ by $\Psi(e):=(\pi_1 e,\pi_2 e,\ldots ,\pi_n e)$ for all $e
\in E$. Then $\Psi$ is also a GPEA-morphism and by Theorem
\ref{th:FinitePwiseDisjointPi} (ii), $\Phi\circ\Psi$ is the identity on $E$.
Now consider $\pi_i e_j$ for $i,j=1,2,\ldots ,n$. We have $\pi_i e_j=\pi_i
(\pi_j e_j)=(\pi_i\wedge\pi_j)e_j$. Thus $\pi_i e_j=0$ for $i\not=j$ and
$\pi_i e_j=e_j$ for $i=j$ and so $\Psi\circ\Phi$ is the identity on $X$.
Consequently $\Psi=\Phi^{-1}$ and $\Phi$ is a GPEA-isomorphism.
\end{proof}

According to the previous theorem, we may consider $E$ as a direct sum
$E=\pi_1(E)\oplus\pi_2(E)\oplus\dots\oplus\pi_n(E)$ whenever $\pi_i$ are
pairwise disjoint elements of $\GEX(E)$ and $\bigvee_{i=1}^n\pi_i=1$. In
particular, $E=\pi(E)\oplus\pi\,'(E)$ for every $\pi$ in the boolean algebra
$\GEX(E)$.

\begin{theorem} \label{th:CentId=piE}
If $S\subseteq E$, then the following statements are equivalent:
\begin{itemize}
\item[(i)] $S$ is a central ideal {\rm(}direct summand{\rm)} of $E$.
\item[(ii)] There exists $\pi\in\GEX(E)$ such that $S=\pi(E)$.
\end{itemize}
\end{theorem}

\begin{proof}
Assume that $E=S\oplus S'$. We define for $e\in E$: $\pi e:=s$ where $e=
s\oplus t$, $s\in S$, $t\in S'$. Then $\pi\in\GEX(E)$. Indeed, (EXC1) and
(EXC2) hold trivially and since $s\leq s\oplus t$, (EXC3) also holds. If
$e,f\in E$ are such that $\pi e=e$ and $\pi f=0$, then $e\in S$ and $f\in S'$,
thus $e\perp f$ and so (EXC4) holds too. If, on the other hand, $\pi\in
\GEX(E)$ and $S=\pi(E)$, then $E=S\oplus\pi\,'(E)$ and so $S$ is a central
ideal.
\end{proof}

\begin{corollary} \label{co:piEnormal}
If $\pi\in\GEX(E)$, then $\pi(E)$ is a normal ideal in $E$.
\end{corollary}

\begin{proof}
By Theorem \ref{th:CentId=piE}, $\pi(E)$ is a central ideal in $E$,
and by Proposition \ref{pr:DirSumNormal}, every central ideal in $E$
is normal.
\end{proof}

\begin{corollary}\label{co:CIposet}
Let us partially order the set $C$ of all central ideals {\rm(}direct
summands{\rm)} of $E$ by inclusion. Then there is an order isomorphism
between $\GEX(E)$ and $C$ given by: $\pi\leftrightarrow S$ iff $\pi(E)=S$.
Moreover, if $\pi(E)=S$, then $\pi\,'(E)$ is the direct summand $S'$ of $E$
that is complementary to $S$.
\end{corollary}

\begin{theorem}\label{th:PtwisePi}
Let $\pi\in\GEX(E)$ and let $(e_i)_{i\in I}$ be a family of elements in $E$.
Then{\rm:}
\begin{itemize}
\item[(i)] If $\bigvee_{i\in I} e_i$ exists in $E$, then so does $\bigvee_
 {i\in I} \pi e_i$ and $\pi(\bigvee_{i\in I} e_i)=\bigvee_{i\in I} \pi e_i$.
\item[(ii)] If $I\not=\emptyset$ and $\bigwedge_{i\in I} e_i$ exists in $E$,
 then so does $\bigwedge_{i\in I} \pi e_i$ and $\pi(\bigwedge_{i\in I} e_i)=
 \bigwedge_{i\in I} \pi e_i$.
\item[(iii)] If $(e_i)_{i\in I}$ is orthosummable, then so is $(\pi e_i)
 _{i\in I}$ and $\pi(\oplus_{i\in I} e_i)=\oplus_{i\in I}\pi e_i$.
\end{itemize}
\end{theorem}

\begin{proof}
(i) Put $e:=\bigvee_{i\in I} e_i$. As $e_i\leq e$, we also have $\pi e_i
\leq\pi e$ for all $i\in I$. Now suppose that $\pi e_i\leq f$ for all $i\in I$.
Then $\forall i\in I$: $\pi e_i=\pi (\pi e_i)\leq\pi f$. But we also have
$\pi\,' e_i\leq\pi\,' e$ for all $i\in I$. So by (vii) and (viii) in Theorem
\ref{th:EXCprop}, $e_i=\pi e_i\oplus\pi\,' e_i=\pi e_i\vee\pi\,' e_i\leq
\pi f\vee\pi\,' e=\pi f\oplus\pi\,' e$ for all $i\in I$. Thus $e\leq\pi f
\oplus\pi\,' e$ so $\pi e\leq\pi f\oplus\pi (\pi\,' e)=\pi f\leq f$. Hence
$\pi e=\bigvee_{i\in I} \pi e_i$.

(ii) Put $e:=\bigwedge_{i\in I} e_i$. As $e\leq e_i$, we have $\pi e
\leq\pi e_i$ for all $i\in I$. Suppose $f\in E$ with $f\leq\pi e_i$ for
all $i\in I$. As $I\not=\emptyset$, Theorem \ref{th:EXCprop} (iii) implies
that $f=\pi f$. Because $\pi e_i\leq e_i$, we have $f\leq e_i$ for all
$i\in I$. Therefore $f\leq e$ and $\pi f=f\leq \pi e$.

(iii) For any finite subset $F$ of $I$, as $\pi$ is a GPEA-endomorphism,
$\pi(\oplus_{i\in F}e_i)=\oplus_{i\in F}\pi e_i$. As $\oplus_{i\in I}
\pi e_i=\bigvee_{F}\oplus_{i\in F}\pi e_i=\bigvee_{F}\pi(\oplus_{i\in F}
e_i)=\pi\bigvee_{F} (\oplus_{i\in F} e_i)=\pi(\bigvee_{i\in I} e_i)$, the
desired result follows from (i).
\end{proof}

\section{The center of a GPEA} \label{sc:CenterGPEA} 

\begin{definition}\label{df:Center}
An element $c\in E$ is \emph{central} iff for every $a, b\in E$, the
following hold:
\begin{itemize}
\item[(C1)] There exist $a_1, a_2\in E$ such that $a_1\leq c$, $a_2\oplus c$
 exists and $a=a_1\oplus a_2$.
\item[(C2)] If $a\leq c$ and if $b\oplus c$ exists, then $a\perp b$.
\item[(C3)] If $a,b\leq c$ and $a\oplus b$ exists, then $a\oplus b\leq c$.
\item[(C4)] If $a\oplus c$, $b\oplus c$ and $a\oplus b$ exist, then $a
 \oplus b\oplus c$ exists.
\end{itemize}
We denote the set of all central elements of the GPEA $E$ by
$\Gamma(E)$.
\end{definition}

\begin{lemma}\label{le:CentProp}
Let $a,x,y\in E$ and let $c\in\Gamma(E)$. Then{\rm:}
\begin{enumerate}
\item The elements $a_1$ and $a_2$ in {\rm(}C1{\rm)} of Definition
 \ref{df:Center} are unique and $a_1\perp a_2$.
\item $\forall\, a\in E$, $a\oplus c$ exists iff $a\perp c$ iff
 $c\oplus a$ exists.
\item If $x\oplus y$ exists in $E$ and at least one of the elements
 $x$, $y$ is central, then $x\perp y$.
\end{enumerate}
\end{lemma}

\begin{proof}
(i) Suppose that $a\sb{1},a\sb{2},b\sb{1},b\sb{2}\in E$ with
$a=a\sb{1}\oplus a\sb{2}=b\sb{1}\oplus b\sb{2}$, where $a_1,b_1
\leq c$ and both $a_2\oplus c$ and $b_2\oplus c$ exist. Then by
(C2), we have $a_1\perp a_2$ and $b_1\perp b_2$, whence $a=a_1
\oplus a_2=a_2\oplus a_1=b_1\oplus b_2=b_2\oplus b_1$. As $a\sb{1}\leq
c$, there exists $d\in E$ such that $a\sb{1}\oplus d=c$ and we have
$a\sb{2}\oplus c=a\sb{2}\oplus a\sb{1}\oplus d=b\sb{2}\oplus b
\sb{1}\oplus d$. Since $b\sb{1},d\leq c$, (C3) implies that $b
\sb{1}\oplus d\leq c$, whence $a\sb{2}\oplus c\leq b\sb{2}\oplus c$,
and it follows by cancellation that $a\sb{2}\leq b\sb{2}$. By symmetry,
$b\sb{2}\leq a\sb{2}$, so $a\sb{2}=b\sb{2}$, and therefore $a\sb{1}
=b\sb{1}$ by cancellation.

(ii) If $a\oplus c$ exists, then as $c\leq c$, we have $a\perp c$ by
(C2). As $a\perp c$, then $c\oplus a$ exists. Finally, suppose that
$c\oplus a$ exists. Then by (C1), there exist $d\sb{1},d\sb{2}\in E$
with $c\oplus a=d\sb{1}\oplus d\sb{2}$, where $d\sb{1}\leq c$ and
$d\sb{2}\oplus c$ exists. As $c\leq c$ and $d\sb{2}\oplus c$ exists,
(C2) implies that $c\perp d\sb{2}$. Also, by part (i), $d\sb{1}\perp d
\sb{2}$, and since $d\sb{1}\leq c$, we have $c\oplus a=d\sb{1}\oplus d
\sb{2}=d\sb{2}\oplus d\sb{1}\leq d\sb{2}\oplus c=c\oplus d\sb{2}$, whence
$a\leq d\sb{2}$ by cancellation. Thus, $a\leq d\sb{2}$ and $d\sb{2}\perp c$,
so $a\oplus c$ exists by Lemma \ref{le:SlashProps} (iii). Part (iii) follows
immediately from (ii).
\end{proof}

\begin{theorem}\label{th:centr}
If $c\in E$, then the following are equivalent{\rm}:
\begin{itemize}
\item[(i)] $c$ is central, i.e., $c\in\Gamma(E)$.
\item[(ii)] $E[0,c]$ is a central ideal {\rm(}direct summand{\rm)}
 of $E$.
\item[(iii)] $E$ decomposes as a direct sum $E=E[0,c]\oplus\{f\in E:
 f\perp c\}$.
\end{itemize}
\end{theorem}

\begin{proof}
(i) $\Rightarrow$ (ii): If $c$ is central, then by (C3), $E[0,c]$ is an
ideal. We prove that it is moreover a central ideal; that is, there
exists another ideal, namely $E[0,c]':=\{e\in E: e\perp c\}$, such that
$E=E[0,c]\oplus E[0,c]'$. By Definition \ref{df:Center} and Lemma \ref
{le:CentProp} (ii), for every $e\in E$ there exist $e_1,e_2\in E$ such
that $e=e_1\oplus e_2$, where $e_1\in E[0,c]$ and $e_2\in E[0,c]'$. It
will be sufficient to show that $E[0,c]'$ is an ideal in $E$. If $d\leq e$
and $e\in E[0,c]'$, then by Lemma \ref{le:SlashProps} (iii), $d\oplus c$
exists; whence, as $c\in\Gamma(E)$, we have $d\in E[0,c]'$. Finally,
suppose that $e,f\in E[0,c]'$ and $e\oplus f$ exists. Then by (C4), $e
\oplus f\oplus c$ exists, and again, as $c\in\Gamma(E)$, it follows that
$e\oplus f\in E[0,c]'$.

(ii) $\Rightarrow$ (iii): If $E[0,c]$ is a central ideal in $E$, then there
is an ideal $E[0,c]'$ such that $E=E[0,c]\oplus E[0,c]'$. Evidently, if $f
\in E[0,c]'$, then $f\perp c$. Conversely, if $f\in E$ with $f\perp c$, then
$f=s\oplus t$, where $s\leq c$ and $t\in E[0,c]'$. As $s\leq f$ and $f
\perp c$, we get $s\perp c$ and since $E[0,c]$ is an ideal, $s\oplus c
\leq c$, which entails $s=0$. Thus $f=t\in E[0,c]'$ and $E[0,c]'=\{f\in E:
f\perp c\}$.

(iii) $\Rightarrow$ (i): Let $E=E[0,c]\oplus\{f\in E: f\perp c\}$. We prove
(C1)--(C4). (C1) follows directly from the fact, that every $e\in E$ can be
written as $e=e_1\oplus e_2$, where $e_1\in E[0,c]$ and $e_2\in\{f\in E:
f\perp c\}$. To prove (C2), suppose that $a\leq c$ and $b\oplus c$ exists.
Then, we can write $b\oplus c=e\sb{1}\oplus f\sb{1}$ where $e\sb{1}\in
E[0,c]$, $f\sb{1}\in\{f\in E: f\perp c\}$, and $e\sb{1}\perp f\sb{1}$.
Therefore, $b\oplus c=f\sb{1}\oplus e\sb{1}\leq f\sb{1}\oplus c$, so
$b\leq f\sb{1}$ by cancellation; hence, since $\{f\in E: f\perp c\}$ is
an ideal, it follows that $b\in\{f\in E: f\perp c\}$. Now we have $a\in
E[0,c]$ and $b\in\{f\in E:f\perp c\}$, whence $a\perp b$, proving (C2).
Because $E[0,c]$ is an ideal, (C3) follows immediately. For (C4), suppose
$a\oplus c$, $b\oplus c$ and $a\oplus b$ all exist. As a consequence of (C2)
and the fact that $c\leq c$, we have $a\perp c$ and $b\perp c$, i.e.,
$a,b\in\{f\in E:f\perp c\}$.  Again, since $\{f\in E:f\perp c\}$ is an ideal,
we infer that $a\oplus b\perp c$, so $a\oplus b\oplus c$ exists, proving
(C4).
\end{proof}

\begin{definition} \label{df:pisbc}
If $c\in \Gamma(E)$, then by Theorems \ref{th:centr} and \ref{th:CentId=piE},
there exists uniquely determined mapping in $\GEX(E)$, henceforth denoted by
$\pi_c$, such that $\pi_c(E)=E[0,c]$.
\end{definition}

\begin{corollary}\label{co:pic}
Let $\pi\in\GEX(E)$. Then the following statements are equivalent:
\begin{enumerate}
\item There exists a largest element $c\in\pi(E)$.
\item $\pi(E)=E[0,c]$.
\item $c\in\Gamma(E)$, $\pi=\pi_c$, and $\pi\,'(E)=\{f\in E:f\perp c\}$.
\end{enumerate}
\end{corollary}

\begin{proof}
Since $\pi(E)$ is an ideal in $E$, $\pi(E)=E[0,c]$ iff $c$ is the
largest element in $\pi(E)$. The rest follows by Theorem \ref
{th:CentId=piE}, Theorem \ref{th:centr}, and Definition \ref{df:pisbc}.
\end{proof}

If $c\in\Gamma(E)$ and $d\in E$ with $c\leq d$, then there exists
$x:=d\backslash c\in E$ with $x\oplus c=d$, and since $c\in\Gamma(E)$,
it follows from Lemma \ref{le:CentProp} (iii) that $x\perp c$, whence
$c\oplus x=d$ also holds, i.e., $x=c/d$. Consequently, $d\ominus c=
d\backslash c=c/d$ exists (Definition \ref{df:leqetc}). In particular,
$d\ominus c$ is defined for $c,d\in\Gamma(E)$ iff $c\leq d$, and if
$c\leq d$, then by part (x) of the next theorem, $d\ominus c\in\Gamma(E)$
and we have $d=c\oplus(d\ominus c)=(d\ominus c)\oplus c$.

We omit the proofs of the following two theorems as they can be obtained
by easy modifications of the proofs of \cite[Lemma 4.5, Theorem 4.6]{ExoCen}.

\begin{theorem}\label{th:ceprop}
Let $c,d\in\Gamma(E)$, $e\in E$. Then{\rm:}
\begin{enumerate}
\item $\pi_c e=e\wedge c$.
\item $\pi_c d=\pi_d c=c\wedge d$.
\item $e\wedge c=0\,\Leftrightarrow\, e\in(\pi_c)'(E)\,\Leftrightarrow\,
 e\perp c$.
\item $c\wedge d\in\Gamma(E)$ and $\pi_{c\wedge d}=\pi_c\wedge\pi_d$.
\item $c\wedge d=0\,\Leftrightarrow\,\pi_c\wedge\pi_d=0\,\Leftrightarrow\,
 c\perp d$.
\item If $c\perp d$, then $c\oplus d=c\vee d\in\Gamma(E)$ and $\pi_{c
 \oplus d}=\pi_{c\vee d}=\pi_c\vee\pi_d$.
\item $\pi_c$ is the smallest $\pi\in\GEX(E)$ such that $\pi c=c$.
\item If $\pi\in\GEX(E)$ and $h\in E$, then $h\in\Gamma(E)$ iff $\pi e
 =e\wedge h$ for all $e\in E$, and in this case, $\pi=\pi_h$.
\item $c\leq d\,\Leftrightarrow\,\pi_c\leq\pi_d$.
\item If $c\leq d$, then $d\ominus c$ exists, $d\ominus c\in\Gamma(E)$
 and $\pi_{d\ominus c}=\pi_d\wedge(\pi_c)'$.
\item $c\vee d$ exists in $E$, $c\vee d\in\Gamma(E)$ and $\pi_{c\vee d}
 =\pi_c\vee\pi_d$.
\end{enumerate}
\end{theorem}

\begin{theorem}\label{th:centgea}
\rm{(i)} $\{\pi_c: c\in\Gamma(E)\}$ is a sublattice of the boolean algebra
$\GEX(E)$, and as such, it is a generalized boolean algebra. \rm{(ii)}
$\Gamma(E)$ is a commutative lattice-ordered sub-GPEA (hence sub-GEA) of $E$.
\rm{(iii)} The mapping $c\mapsto\pi_c$ from $\Gamma(E)$ onto $\{\pi_c: c\in
\Gamma(E)\}$ is a lattice isomorphism. \rm{(iv)} $\Gamma(E)$ is a generalized
boolean algebra, i.e., a distributive and relatively complemented lattice with
smallest element $0$. \rm{(v)} $E$ is a PEA iff $\{\pi_c: c\in\Gamma(E)\}=\GEX(E)$.
\end{theorem}

If $\phi$ is a mapping defined on $E$ and $S\subseteq E$, then $\phi|\sb{S}$
denotes the restriction of $\phi$ to $S$. The proofs of parts (i)--(iv) of
the next theorem are easy modifications of the proofs of \cite[Theorem 4.13,
(i)--(iv)]{CenGEA}; part (v) follows as in the proof of \cite[Lemma 4.5 (iii)]
{CenGEA}; and with the aid of part (v), part (vi) follows as in the proof of
\cite[Theorem 4.13 (v)]{CenGEA}.

\begin{theorem}\label{th:mis}
Let $\xi, \pi \in \GEX(E)$. Then{\rm:}
\begin{enumerate}
\item $\xi|\sb{\pi(E)}\in \GEX(\pi(E))$.
\item If $\tau \in \GEX(\pi(E))$, then $\tau\circ\pi\in\GEX(E)$.
\item $\xi\mapsto\xi|\sb{\pi(E)}$ is a surjective boolean homomorphism of
 $\GEX(E)$ onto $\GEX(\pi(E))$.
\item If $p\in \pi(E)$, then $\pi(E)[0,p]$ and $E[0,p]$ coincide both
 as sets and as pseudoeffect algebras.
\item If $p\in E$, then $\pi(E[0,p])=E[0,\pi p]=\pi(E)[0,\pi p]$.
\item $\Gamma(\pi(E))=\Gamma(E)\cap\pi(E)$.
\end{enumerate}
\end{theorem}

\begin{lemma}\label{le:nova}
If $\pi\in\GEX(E)$ and $k\in E$, then $\pi|\sb{E[0,k]}\in\GEX(E[0,k])$.
\end{lemma}

\begin{proof}
We prove that $\pi|\sb{E[0,k]}$ satisfies (EXC1)--(EXC4) for the PEA
$E[0,k]$. Let $a,b\in E[0,k]$. We have $\pi|\sb{E[0,k]}a=\pi a\leq a
\leq k$, so $\pi|\sb{E[0,k]}\colon E[0,k]\to E[0,k]$. To prove (EXC1),
suppose that $a\oplus\sb{k}b=a\oplus b\leq k$. Then $\pi|\sb{E[0,k]}
(a\oplus\sb{k}b)=\pi(a\oplus b)=\pi(a)\oplus\pi(b)\leq a\oplus b\leq k$,
so $\pi|\sb{E[0,k]}$ is a GPEA-endomorphism of $E[0,k]$. Conditions
(EXC2) and (EXC3) hold trivially. To prove (EXC4), suppose that $\pi|
\sb{E[0,k]}a=\pi a=a$ and $\pi|\sb{E[0,k]}b=\pi b=0$. Then $a\perp b$,
so $a\oplus b=b\oplus a$. Also $\pi\,'b=b$, and by Lemma
\ref{le:DisjointPiXi} (i) with $\xi:=\pi\,'$, $a\oplus b=b\oplus a=
a\vee b\leq k$. Therefore, $a\oplus\sb{k}b=a\oplus b=b\oplus a=b
\oplus\sb{k}a$, i.e., $a$ is orthogonal to $b$ in $E[0,k]$, proving
(EXC4).
\end{proof}

\section{Central orthocompleteness} \label{sc:CO} 

\begin{definition} \label{df:GammaexOrghogonal}
We say that elements $e,f\in E$ are \emph{$\GEX$-orthogonal} iff there
are $\pi,\xi\in\GEX(E)$ such that $\pi\wedge\xi=0$, $\pi e=e$ and
$\xi f=f$. More generally, an arbitrary family $(e_i)_{i\in I}$ in $E$
is $\GEX$-orthogonal iff there is a pairwise disjoint family $(\pi_i)
_{i\in I}$ in $\GEX(E)$ such that $\pi_i e_i=e_i$ for all $i\in I$.
\end{definition}

As is easily seen, elements $e,f\in E$ are $\GEX$-orthogonal iff there
is a direct sum decomposition $E=S\oplus S'$ such that $e\in S$ and
$f\in S'$.

\begin{lemma}\label{le:co}
\rm{(i)} A finite family $(e_i)_{i=1}^n$ in $E$ is pairwise
$\GEX$-orthogonal iff it is $\GEX$-orthogonal and then it is orthogonal
with $\oplus_{i=1}^n e_i=\bigvee_{i=1}^n e_i$. \rm{(ii)} If an arbitrary
family $(e_i)_{i\in I}\in E$ is $\GEX$-orthogonal, then it is orthogonal
and it is orthosummable iff its supremum exists in $E$, in which case
$\oplus_{i\in I} e_i=\bigvee_{i\in I} e_i$.
\end{lemma}

\begin{proof}
(i) Clearly, a subfamily of a $\GEX$-orthogonal family is $\GEX$-orthogonal.
It is also clear from the definition, that every $\GEX$-orthogonal family is
pairwise $\GEX$-orthogonal. We prove both the converse and orthogonality by
induction on $n$. For $n=1$ the assertion obviously holds. Suppose now the
statement holds for $(n-1)$ elements, $n>1$ and assume that $(e_i)_{i=1}^n$
is a pairwise $\GEX$-orthogonal family. Then by the induction hypotheses,
$(e_i)_{i=1}^{n-1}$ is orthogonal, $\oplus_{i=1}^{n-1} e_i=\bigvee_{i=1}
^{n-1}e_i$, and there exist pairwise disjoint mappings $\xi_i\in\GEX(E)$
with $\xi_i e_i=e_i$ for $i=1,2,\ldots ,n-1$. Moreover, $e_i$ and $e_n$
are $\GEX$-orthogonal for $i=1,2,\ldots,n-1$; hence there exist $\alpha_i,
\beta_i\in\GEX(E)$ with $\alpha_i\wedge\beta_i=0$, $\alpha_i e_i=e_i$, and
$\beta_i e_n=e_n$. For $i=1,2,\ldots ,n-1$, put $\pi_i:=\xi_i\wedge\alpha_i$
and put $\pi_n:=\bigwedge_{i=1}^{n-1}\beta_i$. Then $\pi_i\in\GEX(E)$ are
pairwise disjoint and $\pi_i e_i=e_i$ for $i=1,2,\ldots ,n$, so the family
$(e_i)_{i=1}^n$ is $\GEX$-orthogonal. We now put $\pi:=\bigvee_{i=1}^{n-1}
\pi_i$ to get $\pi\wedge\pi_n=0$, $\pi(\oplus_{i=1}^{n-1} e_i)=\oplus_
{i=1}^{n-1}\pi e_i=\oplus_{i=1}^{n-1} e_i$, and $\pi_n e_n=e_n$; hence
by Lemma \ref{le:DisjointPiXi} (i), $(\oplus_{i=1}^{n-1} e_i)\perp e_n$
and $\oplus_{i=1}^n e_i=(\oplus_{i=1}^{n-1} e_i)\oplus e_n=(\vee_{i=1}^{n-1}
e_i)\vee e_n=\vee_{i=1}^n e_i$.

(ii) If $(e_i)_{i\in I}$ is $\GEX$-orthogonal, then every finite subfamily
is $\GEX$-orthogonal and by (i), $\oplus_{i\in F}e_i=\bigvee_{i\in F}e_i$,
where $F$ is any finite subset of $I$. Therefore $\bigvee_{i\in I}e_i=
\bigvee_F(\bigvee_{i\in F}e_i)=\bigvee_F(\oplus_{i\in F} e_i)=\oplus_
{i\in I}e_i$.
\end{proof}

\begin{lemma}\label{le:coce}
\rm{(i)} $c,d\in\Gamma (E)$ are $\GEX$-orthogonal iff $\pi_c\wedge\pi_d =0$
iff $c\perp d$ iff $c\wedge d=0$.
\rm{(ii)} A family of central elements is $\GEX$-orthogonal iff it is
orthogonal iff it is pairwise orthogonal iff it is pairwise disjoint.
\end{lemma}

\begin{proof}
(i) If $\pi_c\wedge\pi_d=0$, then $c$ and $d$ are $\GEX$-orthogonal by definition.
If $c,d$ are $\GEX$-orthogonal, then there exist $\pi,\xi\in\GEX(E)$ such that
$\pi c=c$, $\xi d=d$ and $\pi\wedge\xi=0$. But $\pi_c\leq\pi$ and $\pi_d\leq\xi$
by Theorem \ref{th:centgea} (vii), thus $\pi_c$ and $\pi_d$ are disjoint too.
The remaining equivalences follow from Theorem \ref{th:ceprop} (v).

(ii) If the family $(c_i)_{i\in I}$ of central elements in $E$ is
$\GEX$-orthogonal, then by Lemma \ref{le:co} (ii) it is orthogonal. If it is
orthogonal, then by the definition of orthogonality it is pairwise orthogonal.
If it is pairwise orthogonal, then by Theorem \ref{th:ceprop} (v) it is
pairwise disjoint. Finally, suppose that $(c_i)_{i\in I}$ is pairwise
disjoint. Then by Theorem \ref{th:ceprop} (v) again, $(\pi\sb{c\sb{i}})
\sb{i\in I}$ is a pairwise disjoint family in $\GEX(E)$ such that
$\pi\sb{c\sb{i}}c\sb{i}=c\sb{i}$ for all $i\in I$, so $(c_i)_{i\in I}$
is $\GEX$-orthogonal.
\end{proof}

\begin{definition}\label{de:COGPEA}
The generalized pseudo-effect algebra $E$ is \emph{centrally orthocomplete}
(COGPEA) iff it satisfies the following conditions:
\begin{itemize}
\item[(CO1)] Every $\GEX$-orthogonal family in $E$ is orthosummable, i.e.
 (Lemma \ref{le:co} (ii)), it has a supremum) in $E$.
\item[(CO2)] If $e\in E$ is such that $e\oplus e_i$ (resp. $e_i\oplus e)$
 exists for every element of a $\GEX$-orthogonal family $(e_i)_{i\in I}
 \subset E$, then $e\oplus(\oplus_{i\in I} e_i)$ (resp. $(\oplus_{i
 \in I}e_i)\oplus e$) exists in $E$.
\end{itemize}
\end{definition}

\begin{theorem}\label{th:COGPEAp'wisedisj}
Let $E$ be a COGPEA and $(\pi_i)_{i\in I}$ a pairwise disjoint family in
$\GEX(E)$. Let $(e_i)_{i\in I}$, $(f_i)_{i\in I}$ be families of elements
in $E$ such that $e_i\oplus f_i$ exists for all $i\in I$ and $e_i,f_i\in
\pi_i(E)$. Then{\rm:}
\begin{enumerate}
\item $(e_i)_{i\in I}$, $(f_i)_{i\in I}$, and $(e_i\oplus f_i)_{i\in I}$
 are $\GEX$-orthogonal, hence orthosummable.
\item $\oplus_{i\in I} e_i=\bigvee_{i\in I} e_i$, $\oplus_{i\in I} f_i=
 \bigvee_{i\in I} f_i$ and $\oplus_{i\in I}(e_i \oplus f_i)=\bigvee_{i\in I}
 (e_i\oplus f_i)$.
\item $(\oplus_{i\in I} e_i)\oplus(\oplus_{i\in I} f_i)$ exists.
\item $(\oplus_{i\in I} e_i)\oplus(\oplus_{i\in I} f_i)=\oplus_{i\in I}
 (e_i\oplus f_i)=\bigvee_{i\in I}(e_i\oplus f_i)$.
\end{enumerate}
\end{theorem}

\begin{proof}
Since $e_i$, $f_i$ belong to $\pi_i(E)$ for every $i\in I$, so does $e_i
\oplus f_i$. Thus (i) follows directly from (CO1) and the definition of
$\GEX$-orthogonality, and (ii) is implied by Lemma \ref{le:co} (ii).

(iii) Put $e:=\oplus_{i\in I} e_i=\bigvee_{i\in I} e_i$ and $f:=
\oplus_{i\in I}f_i=\bigvee_{i\in I} f_i$. By hypotheses $e_i\oplus f_i$
exists for every $i\in I$, and for $i\not=j$, $e_i\oplus f_j$ also exists
by Lemma \ref{le:DisjointPiXi} (i). Applying (CO2) we find that $e_i\oplus f$
exists for all $i\in I$, and applying (CO2) once more we conclude that
$e\oplus f$ exists too.

(iv) As $e\oplus f$ exists, so does $e\sb{i}\oplus f$ for every $i\in I$,
and therefore by Lemma \ref{le:oplusdist}, $e\sb{i}\oplus f=\bigvee\sb{j\in I}
(e\sb{i}\oplus f\sb{j})$. Therefore a second application of Lemma
\ref{le:oplusdist} yields
\setcounter{equation}{0}
\begin{equation} \label{eq:p'wisedisj01}
e\oplus f=(\bigvee\sb{i\in I}e_i)\oplus f=\bigvee\sb{i\in I}(e\sb{i}
 \oplus f)=\bigvee\sb{i\in I}\bigvee\sb{j\in I}(e\sb{i}\oplus f\sb{j})
 =\bigvee\sb{i,j\in I}(e\sb{i}\oplus f\sb{j}).
\end{equation}
Also, by Lemma \ref{le:DisjointPiXi} (i), for all $i,j\in I$,
\begin{equation} \label{eq:p'wisedisj02}
i\not=j\Rightarrow e\sb{i}\oplus f\sb{j}=e\sb{i}\vee f\sb{j}
 \leq(e\sb{i}\oplus f\sb{i})\vee(e\sb{j}\oplus f\sb{j})\leq
 \bigvee\sb{i\in I}(e\sb{i}\oplus f\sb{i}).
\end{equation}
Combining (\ref{eq:p'wisedisj01}) and (\ref{eq:p'wisedisj02}), and using
(ii) above, we conclude that $e\oplus f=\bigvee\sb{i\in I}(e\sb{i}\oplus f
\sb{i})=\oplus\sb{i\in I}(e\sb{i}\oplus f\sb{i})$.
\end{proof}

\begin{theorem}\label{th:DisjSup}
If $E$ is a COGPEA and $(\pi_i)_{i\in I}$ is a pairwise disjoint family
of elements in $\GEX(E)$, then the supremum $\bigvee_{i\in I} \pi_i$ exists
in the boolean algebra $\GEX(E)$ and for every $e\in E$, \ $(\bigvee
\sb{i\in I}\pi\sb{i})e=\bigvee\sb{i\in I}\pi\sb{i}e=\oplus\sb{i\in I}
\pi\sb{i}e$.
\end{theorem}

\begin{proof}
Let $e,f\in E$ and $i,j\in I$. The family $(\pi_i)_{i\in I}$ is pairwise
disjoint and $\pi_i(\pi_i e)=\pi_i e$ for every $i\in I$, whence $(\pi_i e)
_{i\in I}$ is a $\GEX$-orthogonal family in $E$. Thus by (CO1) $(\pi_i e)
_{i\in I}$ is orthosummable with $\oplus_{i\in I}\pi_i e=\bigvee_{i\in I}
\pi_i e$ (Lemma \ref{le:co} (ii)). We define $\pi\colon E\to E$ by $\pi e
:=\bigvee_{i\in I}\pi_i e=\oplus_{i\in I}\pi_i e$. It will be sufficient to
prove that $\pi$ is in $\GEX(E)$ and that it is the supremum of $(\pi_i)_{i\in I}$
in $\GEX(E)$.

Suppose $e\oplus f$ exists, so that $\pi\sb{i}(e\oplus f)=\pi\sb{i}e
\oplus\pi\sb{i}f$ for all $i\in I$. In Theorem \ref{th:COGPEAp'wisedisj},
put $e\sb{i}:=\pi\sb{i}e$ and $f\sb{i}:=\pi\sb{i}f$ for all $i\in I$ to
infer that $\pi e\oplus \pi f$ exists and
\[
\pi e\oplus \pi f=(\oplus\sb{i}\pi\sb{i}e)\oplus(\oplus\sb{i}\pi\sb{i}f)
 =\oplus\sb{i}(\pi\sb{i}e\oplus\pi\sb{i}f)=\oplus\sb{i}(\pi\sb{i}(e\oplus f))
 =\pi(e\oplus f),
\]
which proves that $\pi$ satisfies (EXC1). We also have $\pi_i(\pi e)=\pi_i
\bigvee_{j\in I}\pi_j e=\bigvee_{j\in I}\pi_i\pi_j e=\pi_i e$ by Theorem
\ref{th:PtwisePi} (i), whence $\pi(\pi e)=\bigvee_{i\in I}\pi_i(\pi e)=
\bigvee_{i\in I}\pi_i e=\pi e$, proving (EXC2). Moreover, as $\pi_i e
\leq e$ for all $i\in I$, it follows that $\pi e=\bigvee_{i\in I}\pi_i e
\leq e$ and therefore (EXC3) holds. To prove (EXC4), suppose that $\pi e
=e$ and $\pi f=0$. Then $\bigvee_{i\in I}\pi_i f=0$, so $\pi_i f=0$ for
all $i\in I$. As $\pi_i(\pi_i e)=\pi_i e$, (EXC4) implies that $\pi_i e
\perp f$ for every $i\in I$. But then, by (CO2), $e=\pi e\perp f$, and
(EXC4) holds for $\pi$ too.

Evidently, $\pi_i e\leq\pi e$ for every $e\in E$, whence $\pi_i\leq\pi$
for all $i\in I$. Also, if $\pi_i\leq\xi\in\GEX(E)$ for all $i\in I$,
then $\pi_i e\leq\xi e$, so $\pi e=\bigvee_{i\in I}\pi_i e\leq\xi e$ for
all $e\in E$ and thus $\pi\leq\xi$. So $\pi=\bigvee_{i\in I}\pi_i$.
\end{proof}

Since a boolean algebra is complete iff every pairwise disjoint subset
has a supremum, Theorem \ref{th:DisjSup} has the following corollary.

\begin{corollary} \label{co:completeboo}
The exocenter $\GEX(E)$ of a COGPEA $E$ is a complete boolean algebra.
\end{corollary}

We may now extend Theorem \ref{th:DisjSup} in the same way as in
\cite[Theorem 6.9]{ExoCen} for an arbitrary family $(\pi_i)_{i\in I}$
in the complete boolean algebra $\GEX(E)$.

\begin{theorem}\label{th:arbp'wisesup}
Suppose that $E$ is a COGPEA, let $(\pi_i)_{i\in I}$ be a family in $\GEX(E)$,
and let $e\in E$. Then{\rm:} {\rm (i)} $\bigvee_{i\in I}\pi_i e$ exists in $E$
and $(\bigvee_{i\in I}\pi_i)e=\bigvee_{i\in I}\pi_i e$.
{\rm (ii)} If $I\not=\emptyset$, then $\bigwedge_{i\in I}\pi_i e$ exists in
$E$ and $(\bigwedge_{i\in I}\pi_i)e=\bigwedge_{i\in I}\pi_i e$.
\end{theorem}

The proof of the next theorem, which extends Theorem \ref{th:finitecartesianprod}
to arbitrary direct sums, is analogous to the proof of \cite[Theorem 6.10]{ExoCen}.

\begin{theorem}\label{th:arbCartProd}
Suppose that $E$ is a COGPEA, let $(\pi_i)_{i\in I}$ be a pairwise disjoint family
in the complete boolean algebra $\GEX(E)$ with $\pi:=\bigvee_{i\in I}\pi_i$, and
consider the cartesian product $X:=${\huge$\times$}$_{i\in I}\pi_i(E)$. Then each
element in $X$ is a $\GEX$-orthogonal {\rm(}hence orthosummable{\rm)} family
$(e_i)_{i\in I}$ and $\oplus_{i\in I} e_i=\bigvee_{i\in I} e_i$. Define the
mapping $\Phi:X\rightarrow\pi(E)$ by $\Phi((e_i)_{i\in I}):=\oplus_{i\in I} e_i$.
Then $\Phi$ is a GPEA-isomorphism of $X$ onto $\pi(E)$ and if $e\in\pi(E)$,
then $\Phi^{-1} e=(\pi_i e)_{i\in I}\in X$.
\end{theorem}

\begin{corollary}\label{cor:cop}
Let $E$ be a COGPEA, let $(p_i)_{i\in I}$ be a nonempty $\GEX$-orthogonal family
in $E$ with $p:=\bigvee_{i\in I}p_i$, let $(\pi_i)_{i\in I}$ be a corresponding
family of pairwise disjoint mappings in $\GEX(E)$ such that $p_i=\pi_i p_i$
for all $i\in I$, and let $X$ be the cartesian product $X:=${\huge$\times$}$
_{i\in I} E[0,p_i]$. Then{\rm: (i)} If $(e_i)_{i\in I}\in X$, then $e_i=\pi_i e
_i$ for all $i\in I$, so $(e_i)_{i\in I}$ is a $\GEX$-orthogonal, hence
orthosummable family in $E$. {\rm(ii)} If $(e_i)_{i\in I}\in X$ with $e:=
\oplus_{i\in I} e_i$, then $\pi_i e=e_i$ for all $i\in I$. In particular,
$\pi_i p=p_i$ for all $i\in I$. {\rm(iii)} If $e\in E[0,p]$, then $\pi_i e
=e\wedge p_i$ for all $i\in I$, $(\pi_i e)_{i\in I}\in X$ and $\bigvee_
{i\in I}\pi_i e=e$. {\rm(iv)} The mapping $\Phi: X\rightarrow E[0,p]$ defined
by $\Phi((e_i)_{i\in I}):=\oplus_{i\in I} e_i=\bigvee_{i\in I} e_i$ is a
PEA-isomorphism of $X$ onto $E[0,p]$ and $\Phi^{-1} (e)=(\pi_i e)_i\in I\in X$
for all $e\in E[0,p]$.
\end{corollary}

The following theorem can also be proved using the same arguments as in the proof
of \cite[Theorem 6.11]{ExoCen}

\begin{theorem}\label{th:COGPEAcenter}
Suppose that $E$ is a COGPEA and $(c_i)_{i\in I}$ is a family of elements in
the center $\Gamma(E)$ of $E$. Then{\rm:} {\rm(i)} If $I\not=\emptyset$, then
$c:=\bigwedge_{i\in I} c_i$ exists in $E$, $c\in\Gamma(E)$, $\pi_c=\bigwedge
_{i\in I}\pi_{c_i}$ and $c$ is the infimum of $(c_i)_{i\in I}$ as calculated
in $\Gamma(E)$. {\rm(ii)} If $(c_i)_{i\in I}$ is bounded above in $E$, then
$d:=\bigvee_{i\in I} c_i$ exists in $E$, $d\in\Gamma(E)$, $\pi_d=\bigvee
_{i\in I}\pi_{c_i}$ and $d$ is the supremum of $(c_i)_{i\in I}$ as calculated
in $\Gamma(E)$.
\end{theorem}

The next theorem extends the results obtained for centrally orthocomplete
GEAs in \cite[Lemma 7.5, Theorem 7.6]{CenGEA}. Here we give a simplified
proof.

\begin{theorem}\label{th:largestandboo}
Let $E$ be a COGPEA. Then{\rm:}
{\rm(i)} There exists a largest element $u\in\Gamma(E)$ and
 $\Gamma(E)\subseteq \pi_u(E)=E[0,u]$.
{\rm(ii)} The center $\Gamma(E)$ is a complete boolean algebra.
\end{theorem}

\begin{proof}
(i) We apply Zorn's lemma to obtain a maximal pairwise disjoint family
of nonzero elements $(c_i)_{i\in I}\subseteq\Gamma(E)$. (Note that
$(c_i)_{i\in I}$ could be the empty family.) By Lemma \ref{le:coce},
$(c_i)_{i\in I}$ is $\GEX$-orthogonal, and since $E$ is a COGPEA,
$u:=\bigvee_{i\in I}c_i=\oplus_{i\in I}c_i$ exists in $E$. Thus
the family $(c_i)_{i\in I}$ is bounded above by $u$ in $E$, and we
infer from Theorem \ref{th:COGPEAcenter} (ii) that $u\in\Gamma(E)$.
Let $c\in \Gamma(E)$. Working in the generalized boolean algebra
$\Gamma(E)$ (Theorem \ref{th:centgea} (iv)), we have $c=(c\wedge u)
\vee d$, where $d:=c\ominus(c\wedge u)\in\Gamma(E)$. As $d\wedge u=0$
and $c\sb{i}\leq u$, it follows that $d\wedge c\sb{i}=0$ for all
$i\in I$, whence $d=0$ by the maximality of $(c_i)_{i\in I}$, and it
follows that $c=c\wedge u\leq u$. Consequently, $\pi\sb{c}\leq\pi\sb{u}$,
and therefore $c\in E[0,c]=\pi_c(E)\subseteq \pi_u(E)=E[0,u]$.

(ii) Since the generalized boolean algebra $\Gamma(E)$ has a unit
(largest element), it is a boolean algebra, and it is complete by
Theorem \ref{th:COGPEAcenter}.
\end{proof}

\begin{theorem}\label{th:centerless}
Let $u$ be the unit {\rm(}largest element{\rm)} in the complete boolean
algebra $\Gamma(E)$ of the COGPEA $E$. Then{\rm:}
\begin{enumerate}
\item The PEA $E[0,u]=\pi_u(E)$ is a direct summand of $E$ and the
 complementary direct summand is $(\pi_u)'(E)=\{f\in E:f\perp u\}=
 \{e\ominus(u\wedge e):e\in E\}$.
\item The center of $E[0,u]$ is $\Gamma(E)$, the complementary direct
 summand $(\pi_u)'(E)$ is centerless {\rm(}i.e., its center is $\{ 0\}${\rm)},
 and no nonzero direct summand of $(\pi_u)'(E)$ is a PEA.
\item If $E=H\oplus K$ where the direct summand $H$ is a PEA and $K$ is
 centerless, then $H=E[0,u]$ and $K=\{ f\in E:f\perp u\}$.
\end{enumerate}
\end{theorem}

\begin{proof} As $u\in\Gamma(E)$, we have $\pi\sb{u}\in\GEX(E)$ as per
Definition \ref{df:pisbc}, by Theorem \ref{th:centr}, the PEA $E[0,u]=
\pi\sb{u}(E)$ is a direct summand of $E$, and its complementary direct
summand is $(\pi_u)'(E)=\{f\in E:f\perp u\}$. If $e\in E$, then
by Theorem \ref{th:ceprop} (i), $\pi\sb{u}e=u\wedge e$, whence $(\pi
\sb{u})'e=\pi\sb{u}e/e=e\backslash\pi\sb{u}e=e\ominus\pi\sb{u}e
=e\ominus(u\wedge e)$, and it follows that $(\pi_u)'(E)=\{e\ominus
(u\wedge e):e\in E\}$.

(ii) As a consequence of Theorem \ref{th:largestandboo} (i), we have
$\Gamma(E)\subseteq \pi_u(E)=E[0,u]$. Therefore, by Theorem \ref{th:mis}
(vi), $\Gamma(E[0,u])=\Gamma(\pi_u(E))=\Gamma(E)\cap\pi_u(E)=\Gamma(E)$.
Also by Theorem \ref{th:mis} (vi), $\Gamma((\pi_u)'(E))=\Gamma(E)\cap
(\pi_u)'(E)\subseteq \pi_u(E)\cap(\pi_u)'(E)=\{0\}$.

(iii) Assume the hypotheses of (iii). By Theorem \ref{th:CentId=piE},
there exists $\pi\in\GEX(E)$ with $\pi(E)=H$, so $K=\pi\,'(E)$. Since
$H$ is a PEA, there is a largest element $c\in H=\pi(E)$; hence by
Corollary \ref{co:pic}, $H=\pi(E)=E[0,c]$, $c\in\Gamma(E)$, $\pi=
\pi\sb{c}$, and $K=(\pi\sb{c})'(E)=\{f\in E:f\perp c\}$. Also, since
$u$ is the largest element in $\Gamma(E)$, we have $c\leq u$, whence
$u\ominus c\in\Gamma(E)$ by Theorem \ref{th:ceprop} (x). Furthermore,
$(u\ominus c)\perp c$, therefore $u\ominus c\in K$, and by Theorem
\ref{th:mis} (vi) we have $u\ominus c\in\Gamma(E)\cap K=\Gamma(E)
\cap(\pi_c)'(E)=\Gamma((\pi_c)'(E))=\Gamma(K)$. Consequently, as
$K$ is centerless, $u\ominus c=0$, so $c=u$, $H=E[0,u]$, and $K=
\{f\in E:f\perp u\}$.
\end{proof}

\section{The exocentral cover} \label{sc:ExoCenCover} 

\begin{definition} \label{df:ExoCenCover}
If $e\in E$, and if there is the smallest mapping in the set $\{\pi\in\GEX(E):
\pi e=e\}$, we will refer to it as \emph{exocentral cover} of $e$ and denote
it by $\gamma_e$. If every element of $E$ has an exocentral cover, we say that
the family $(\gamma_e)_{e\in E}$ is the \emph{exocentral cover system} for $E$,
and in this case, we also denote the set of all mappings in the exocentral
cover system by $\Theta\sb{\gamma}:=\{\gamma\sb{e}:e\in E\}$. (We note that
it is quite possible to have $\gamma\sb{e}=\gamma\sb{f}$ with $e\not=f$.)
\end{definition}

\begin{theorem}
If $E$ is a COGPEA, then the exocentral cover $\gamma\sb{e}$ exists for every
$e\in E$ and $\gamma_e=\bigwedge\{\pi\in\GEX(E):\pi e=e\}\in\GEX(E)$.
\end{theorem}

\begin{proof}
Let $e\in E$ and put $\gamma:=\bigwedge\{\pi:\pi\in\GEX(E),\pi e=e\}$. As
the identity mapping $1$ is in the set $\{\pi\in\GEX(E):\pi e=e\}$, it is
nonempty, and by Theorem \ref{th:arbp'wisesup} (ii),
$$
\gamma e=(\bigwedge\{\pi:\pi\in\GEX(E),\pi e=e\})e=\bigwedge\{\pi e:\pi\in
 \GEX(E),\pi e=e\}=e.
$$
Therefore, $\gamma$ is the smallest mapping in the set $\{\pi\in\GEX(E):
\pi e=e\}$, so $\gamma\sb{e}=\gamma$.
\end{proof}

\begin{theorem}\label{th:ExCovProp}
Let $E$ be a COGPEA and $e,f\in E$. Then{\rm:} {\rm(i)} $\gamma\sb{0}
=0$. {\rm(ii)} $\gamma\sb{e}e=e$. {\rm(iii)} $e\leq f\Rightarrow \gamma_e
\leq\gamma_f$. {\rm(iv)} If $e\oplus f$ exists, then $\gamma_{e\oplus f}
=\gamma_e\vee\gamma_f$. {\rm(v)} $\gamma_{\gamma_e f}=\gamma_e\circ
\gamma_f=\gamma_e\wedge\gamma_f$. {\rm(vi)} $\gamma_{(\gamma_e)'f}=
(\gamma_e)'\circ\gamma_f=(\gamma_e)'\wedge\gamma_f$. {\rm(vii)}
$\gamma\sb{e}\wedge\gamma\sb{f}\in\Theta\sb{\gamma}$. {\rm(viii)}
$(\gamma_e)'\wedge\gamma_f\in\Theta\sb{\gamma}$.
\end{theorem}

\begin{proof}
Parts (i) and (ii) are obvious from Definition \ref{df:ExoCenCover}.

(iii) If $e\leq f=\gamma_f f$, then by Theorem \ref{th:EXCprop} (iii),
$\gamma_f e=e$. But since $\gamma_e$ is the smallest mapping in $\GEX(E)$
that fixes $e$, it follows that $\gamma_e\leq\gamma_f$.

(iv) Suppose that $e\oplus f$ exists. We have $(\gamma_e\vee\gamma_f)e
=\gamma_e e\vee\gamma_f e=e\vee\gamma_f e=e$ because $\gamma_f e\leq e$.
Similarly $(\gamma_e\vee\gamma_f)f=f$. Thus $(\gamma_e\vee\gamma_f)
(e\oplus f)=(\gamma_e\vee\gamma_f)e\oplus(\gamma_e\vee\gamma_f)f=e\oplus f$,
and so $\gamma_{e\oplus f}\leq\gamma_e\vee\gamma_f$. On the other hand,
$e,f\leq e\oplus f$, so by (iii), $\gamma_e,\gamma_f\leq\gamma_{e\oplus f}$
and thus $\gamma_e\vee\gamma_f\leq\gamma_{e\oplus f}$.

(v) Since $\gamma_e\in\GEX(E)$, $\gamma_e(\gamma_e f)=\gamma_e f$ and
$\gamma_f(\gamma_e f)=\gamma_e(\gamma_f f)=\gamma_e f$. Therefore
$\gamma_{\gamma_e f}\leq\gamma_e\wedge\gamma_f=\gamma_e\circ\gamma_f$.
To prove the reverse inequality, consider $f=\gamma_e f\oplus(\gamma_e)' f$
and (iv) to obtain $\gamma_f=\gamma_{\gamma_e f}\vee\gamma_{(\gamma_e)' f}$.
Also $(\gamma_e)'((\gamma_e)'f)=(\gamma_e)'f$, and as $\gamma
\sb{(\gamma\sb{e})'f}$ is the smallest mapping in $\GEX(E)$ that fixes
$(\gamma\sb{e})'f$, we have $\gamma_{(\gamma_e)'f}\leq (\gamma_e)'$. But
then $\gamma_e\wedge\gamma_{(\gamma_e)'f}=0$ and thus $\gamma_e\circ
\gamma_f=\gamma_e\wedge\gamma_f=(\gamma_e\wedge\gamma_{\gamma_e f})
\vee(\gamma_e\wedge\gamma_{(\gamma_e)' f})=\gamma_e\wedge\gamma_{\gamma_e f}
\leq\gamma_{\gamma_e f}$.

(vi) By (v), $(\gamma_e)'\wedge\gamma_{\gamma_e f}=(\gamma_e)'\wedge
\gamma_e\wedge\gamma_f=0$. Also, as in the proof of (v), we have
$\gamma_f=\gamma_{(\gamma_e)'f}\vee\gamma_{\gamma_e f}$ and $\gamma
\sb{(\gamma\sb{e})'f}\leq(\gamma\sb{e})'$. Therefore, $(\gamma_e)'
\wedge\gamma_f=[(\gamma_e)'\wedge\gamma_{(\gamma_e)'f}]\vee[(\gamma_e)'
\wedge\gamma_{\gamma_e f}]=\gamma_{(\gamma_e)'f}\vee 0=\gamma_
{(\gamma_e)'f}$.

Parts (vii) and (viii) follow immediately from parts (v) and (vi).
\end{proof}

\begin{corollary} \label{co:ThetasbgammaGBA}
With the partial order inherited from $\GEX(E)$, $\Theta\sb{\gamma}=
\{\gamma\sb{e}:e\in E\}$ is a generalized boolean algebra.
\end{corollary}

\begin{proof}
By \cite[Theorem 3.2]{HDandTD} with $B:=\GEX(E)$ and $L:=\Theta
\sb{\gamma}$, it will be sufficient to prove that, for all $e,f
\in E$, (i) $\Theta\sb{\gamma}\not=\emptyset$, (ii) $e,f\in E\Rightarrow
(\gamma\sb{e})'\wedge\gamma\sb{f}\in\Theta\sb{\gamma}$, and
(iii) $\gamma\sb{e}\wedge\gamma\sb{f}=0\Rightarrow\gamma\sb{e}
\vee\gamma\sb{f}\in\Theta\sb{\gamma}$. Condition (i) is obvious
and (ii) follows from Theorem \ref{th:ExCovProp} (viii). To prove
(iii), suppose that $\gamma\sb{e}\wedge\gamma\sb{f}=0$. Then, as
$e=\gamma\sb{e}e$ and $f=\gamma\sb{f}f$, Lemma \ref{le:DisjointPiXi}
(i) implies that $e\perp f$; hence by Theorem \ref{th:ExCovProp} (iv),
$\gamma\sb{e}\vee\gamma\sb{f}=\gamma\sb{e\oplus f}\in\Theta\sb{\gamma}$,
proving (iii).
\end{proof}

The following definition, originally formulated for a generalized effect
algebra (GEA) \cite[Definition 7.1]{ExoCen} as a generalization of the
notion of a hull mapping on an effect algebra \cite[Definition 3.1]{HandD},
extends to the GPEA $E$ the notion of a so-called hull system.

\begin{definition} A family $(\eta\sb{e})\sb{e\in E}$ is a \emph{hull
system} for $E$ iff (1) $\eta_0=0$, (2) $e\in E\,\Rightarrow\,\eta_e e=e$,
and (3) $e,f\in E\,\Rightarrow\,\eta_{\eta_e f}=\eta_e\circ\eta_f$. If
$(\eta\sb{e})\sb{e\in E}$ is a hull system for $E$, then an element $e\in E$
is \emph{$\eta$-invariant} iff $\eta\sb{e}f=e\wedge f$ for all $f\in E$.
\end{definition}

\begin{theorem} \label{th:gammahullsys}
If $E$ is a COGPEA, then $(\gamma\sb{e})\sb{e\in E}$ is a hull system for
$E$, the center $\Gamma(E)$ is precisely the set of $\gamma$-invariant
elements in $E$, and for $c\in\Gamma(E)$, $\gamma\sb{c}=\pi\sb{c}$.
\end{theorem}

\begin{proof}
That $(\gamma\sb{e})\sb{e\in E}$ is a hull system for $E$ follows from
parts (i), (ii), and (v) of Theorem \ref{th:ExCovProp}, and the remainder
of the theorem follows from parts (i), (vii), and (viii) of Theorem
\ref{th:ceprop}.
\end{proof}

\begin{theorem} \label{th:disjointgammasbei}
Let $E$ be a COGPEA and $(e_i)_{i\in I}\subseteq E$. Then the family
$(e_i)_{i\in I}$ is $\GEX$-orthogonal iff $\gamma_{e_i}\wedge\gamma_
{e_j}=0$ for all $i,j\in I$, $i\not =j$.
\end{theorem}

\begin{proof}
If $(\gamma\sb{e\sb{i}})\sb{i\in I}$ is pairwise disjoint, then since
$\gamma\sb{e\sb{i}}e\sb{i}=e\sb{i}$, it follows that  $(e\sb{i})\sb{i\in I}$
is $\GEX$-orthogonal. Conversely, suppose that $(e\sb{i})\sb{i\in I}$ is
$\GEX$-orthogonal. Then there exists a pairwise disjoint family $(\pi_i)_
{i\in I}\in\GEX(E)$ such that $\pi_i e_i=e_i$ for all $i\in I$. But then
$\gamma_{e_i}\leq\pi_i$ for all $i\in I$, and therefore the family
$(\gamma_{e_i})_{i\in I}$ is also pairwise disjoint.
\end{proof}

In view of Theorem \ref{th:disjointgammasbei}, a $\GEX$-orthogonal family
of elements of the COGPEA $E$ will also be called \emph{$\gamma$-orthogonal}.

\section{Type determining sets} \label{sc:TDsets} 

\begin{definition} \label{df:fourclosures}
Let $E$ be a COGPEA and $Q,K\subseteq E$. Then we consider four closure operators on the set of all subsets $Q$ of $E$:
\begin{itemize}
\item[(1)] $[Q]_{\gamma}$ is the set of all orthosums (suprema) of $\gamma$-orthogonal families in $Q$, with the understanding that $[\emptyset]_{\gamma}=\{0\}$.
\item[(2)] $Q^{\gamma}:=\{\gamma_e q: e\in E, q\in Q\}$.
\item[(3)] $Q^{\downarrow}:=\bigcup_{q\in Q}E[0,q]$.
\item[(4)] $Q'':=(Q')'$, where $Q':=\{e\in E: q\wedge e=0$ for all $q\in Q\}$.
\end{itemize}
We say that
\begin{itemize}
\item[(5)] $K$ is \emph{type-determining} (TD) set iff $K=[K]_{\gamma}=
 K^{\gamma}$.
\item[(6)] $K$ is \emph{strongly type-determining} (STD) set iff $K=[K]_
 {\gamma}=K^{\downarrow}$.
\end{itemize}
\end{definition}

We note that $Q\subseteq Q''$, $P\subseteq Q\Rightarrow Q'\subseteq P'$,
and $Q'=Q'''$.

\begin{theorem} \label{th:QK}
Let $E$ be a COGPEA and let $Q,K\subseteq E$. Then{\rm:} {\rm(i)} If
$q\in [Q]_{\gamma}$, then there is a $\gamma$-orthogonal family $(q_i)
_{i\in I}$ in $Q$ such that $q=\oplus_{i\in I}q_i=\bigvee_{i\in I}q_i$;
moreover, if $e\leq q$, then $(e\wedge q_i)_{i\in I}$ is a $\gamma$-orthogonal
family in $Q^{\downarrow}$ and $e=$\linebreak$\oplus_{i\in I}(e\wedge q_i)=\bigvee_{i\in I}
(e\wedge q_i)$. {\rm(ii)} $[K^{\gamma}]_{\gamma}$ is the smallest TD subset of
$E$ containing $K$. {\rm(iii)} $[K^{\downarrow}]_{\gamma}$ is the smallest
STD subset of $E$ containing $K$. {\rm(iv)} $K'=(K')^{\downarrow}=
(K^{\downarrow})'$ is STD. {\rm(v)} $K'=([K^{\gamma}]_{\gamma})'=
([K^{\downarrow}]_{\gamma})'$.
\end{theorem}

\begin{proof}
(i) By the definition of $[Q]_{\gamma}$, there exists a family $(q_i)
_{i\in I}$ in $Q$ such that $(\gamma_{q_i})_{i\in I}$ is a pairwise
disjoint family in $\GEX(E)$ and $q=\oplus_{i\in I}q_i=\bigvee_{i\in I}
q_i$. By Theorem \ref{th:PtwisePi} (i), for each $i\in I$, $\gamma_
{q_i}q=\gamma_{q_i}(\bigvee_{j\in I}q_j)=\bigvee_{j\in I}\gamma_{q_i}q_j
=\bigvee_{j\in I}\gamma_{q_i}(\gamma_{q_j}q_j)=q_i$. Therefore, as
$e\leq q$,  we can apply Theorem \ref{th:EXCprop} (iv) to obtain
$\gamma_{q_i} e=e\wedge\gamma\sb{q\sb{i}}q=e\wedge q_i\in Q\sp{\downarrow}$.
By Theorem \ref{th:ExCovProp} (iii), $\gamma_{e\wedge q_i}\leq\gamma_{q_i}$, so
the family $(e\wedge q_i)_{i\in I}$ is $\gamma$-orthogonal. Let us define
$\pi:=\bigvee_{i\in I}\gamma_{q_i}$ in the complete boolean algebra $\GEX(E)$.
Then by Theorem \ref{th:DisjSup}, $\pi q=\bigvee_{i\in I}\gamma_{q_i}q=
\bigvee_{i\in I}q_i=q$, hence, as $e\leq q\in\pi(E)$, it follows by Theorems
\ref{th:EXCprop} (iii) and \ref{th:DisjSup} that $e=\pi e=\bigvee_{i\in I}
\gamma_{q_i}e=\bigvee_{i\in I}(e\wedge q_i)$.

(ii) From the definition it is clear that $[K^{\gamma}]_{\gamma}$ is
contained in every TD set containing $K$. It is also easily seen that
$K\subseteq[K^{\gamma}]_{\gamma}$ and $[[K^{\gamma}]_{\gamma}]_{\gamma}
\subseteq [K^{\gamma}]_{\gamma}$. To prove that $([K^{\gamma}]_{\gamma})^
{\gamma}\subseteq [K^{\gamma}]_{\gamma}$, let $e\in ([K^{\gamma}]_
{\gamma})^{\gamma}$. Then there exists $h\in E$ and $q\in [K^{\gamma}]_
{\gamma}$ with $e=\gamma_h q\leq q$. By (i) with $Q:=K^{\gamma}$, we find
that there exists a $\gamma$-orthogonal family $(q_i)_{i\in I}$ in $K^
{\gamma}$ such that $q=\bigvee_{i\in I} q_i$ and $e=\bigvee_{i\in I}
(e\wedge q_i)$. Thus, as $q_i\leq q$ for all $i\in I$, Theorem \ref
{th:EXCprop} (iv) implies that $\gamma\sb{h}q\sb{i}=q\sb{i}\wedge
\gamma\sb{h}q=q\sb{i}\wedge e$ for all $i\in I$. Also,
as $q_i\in K^{\gamma}$ for every $i\in I$, there exist $h_i\in E$ and
$k_i\in K$ such that $q_i=\gamma_{h_i}k_i$, and we have $e\wedge q\sb{i}=\gamma_h q_i
=\gamma_h\gamma_{h_i}k_i=(\gamma_h\wedge\gamma_{h_i})k_i=\gamma_
{\gamma_h h_i}k_i\in K\sp{\gamma}$. Therefore the elements of the
$\gamma$-orthogonal family $(e\wedge q_i)_{i\in I}$ all belong to
$K^{\gamma}$ and so $e\in [K^{\gamma}]_{\gamma}$.

We omit the proof of (iii) as it is similar to the proof of (ii).

(iv) Evidently, $K'=(K')^{\downarrow}=(K^{\downarrow})'$. It remains to
prove that $[K']\sb{\gamma}\subseteq K'$. Let $q\in [K']_{\gamma}$,
$k\in K$, and $e\in E$ with $e\leq q,k$. By (i) with $Q:=K'$, there are
$\gamma$-orthogonal families $(q_i)_{i\in I}\subseteq K'$ and $(e
\wedge q_i)_{i\in I}$ such that $q=\bigvee_{i\in I} q_i$ and $e=
\bigvee_{i\in I}(e\wedge q_i)$. Since $e\leq k$ and $k\wedge q_i=0$, it
follows that $e\wedge q_i=0$ for all $i\in I$, so $e=0$. Thus $q
\wedge k=0$, whence $q\in K'$.

(v) We have $K\subseteq K''$ and as $K''=(K')'$, it is STD by (iv),
hence it is TD. But then by (ii), $[K^{\gamma}]_{\gamma}\subseteq K''$,
therefore $K'\subseteq ([K^{\gamma}]_{\gamma})'$. We also get
$([K^{\gamma}]_{\gamma})'\subseteq K'$ because $K\subseteq [K^{\gamma}]
_{\gamma}$. Similarly, $K\subseteq [K^{\downarrow}]_{\gamma}$, whence
$([K^{\downarrow}]_{\gamma})'\subseteq K'$ and by (iv) and (iii),
$[K^{\downarrow}]_{\gamma}\subseteq K''$; hence $K'=K'''\subseteq
([K^{\downarrow}]_{\gamma})'$.
\end{proof}

\begin{corollary}
If $A$ {\rm(}which may be empty{\rm)} is the set of all atoms in $E$,
then the STD set $A'$ is the set of all elements in $E$ that dominate
no atom in $E$, and the STD set $A''$ is the set of all elements $p\in E$
such that either $p=0$ or the PEA $E[0,p]$ is atomic.
\end{corollary}

\begin{theorem}\label{centrTD}
The set $\Gamma(E)$ of central elements of a COGPEA $E$ is a TD
subset of $E$.
\end{theorem}

\begin{proof}
Obviously $\Gamma(E)\subseteq [\Gamma(E)]_{\gamma}$ and by theorem
\ref{th:largestandboo} (ii), $[\Gamma(E)]_{\gamma}\subseteq\Gamma(E)$.
To prove that $\Gamma(E)^{\gamma}\subseteq\Gamma(E)$, let $c_1\in
\Gamma(E)^{\gamma}$, so that $c_1:=\gamma_e c$ for some $e\in E$
and $c\in\Gamma(E)$. We claim that $c_1$ is the greatest element of
$\gamma_{c_1}(E)$; hence by Corollary \ref{co:pic}, it is a central
element of $E$. Indeed, if $f\in\gamma\sb{c\sb{1}}(E)$, then $f=\gamma_
{c_1} f=\gamma_{\gamma_e c} f=\gamma_e (\gamma_c f)=\gamma_e(c\wedge f)
\leq\gamma_e c=c_1$ by Theorem \ref{th:ExCovProp} (v) and Theorem \ref{th:ceprop} (i).
\end{proof}

\begin{definition}\label{df:TypeClass}
A nonempty class $\mathcal{K}$ of PEAs is called a \emph{type class} iff
the following conditions are satisfied: (1) $\mathcal{K}$ is closed under
the passage to direct summands. (2) $\mathcal{K}$ is closed under the
formation of arbitrary nonempty direct products. (3) If $E_1$ and $E_2$
are isomorphic PEAs and $E_1$ is in $\mathcal{K}$, then $E_2\in\mathcal{K}$.
If, in addition to (2) and (3), $\mathcal{K}$ satisfies (1$'$) $H\in
\mathcal{K}, h\in H\,\Rightarrow\, H[0,h]\in\mathcal{K}$, then $\mathcal{K}$
is called a \emph{strong type class}.
\end{definition}

\begin{theorem} \label{th:TypeClass}
Let $\mathcal{K}$ be a type class of PEAs and define $K:=\{k\in E:
E[0,k]\in\mathcal{K}\}$. Then $K$ is a TD subset of $E$, and if
$\mathcal{K}$ is a strong type class, then $K$ is STD.
\end{theorem}

\begin{proof}
Suppose $k\in K$ and $e\in E$. Then $E[0,k]\in\mathcal{K}$,
$\gamma_e\in\GEX(E)$, and by Lemma \ref{le:nova}, $\gamma_e|
_{E[0,k]}\in\GEX(E[0,k])$. Thus by Theorem \ref{th:mis} (v)
and Definition \ref{df:TypeClass} (1), $E[0,\gamma_e k]=\gamma_e
(E[0,k])=\gamma_e|_{E[0,k]}(E[0,k])\in\mathcal{K}$, so $K^{\gamma}
\subseteq K$. If $\mathcal{K}$ is a strong type class, it is clear,
that $K^{\downarrow}\subseteq K$. Finally, suppose that $k\in[K]_
{\gamma}$. Then there exists a $\gamma$-orthogonal family $(k_i)_
{i\in I}$ in $K$ such that $k=\bigvee_{i\in I} k_i$. Thus by
Definition \ref{df:TypeClass} (2), $X:=${\huge$\times$}$_{i\in I}
E[0,k_i]\in\mathcal{K}$ and by Corollary \ref{cor:cop}, $X$ is
PEA-isomorphic to $E[0,k]$, whence by Definition \ref{df:TypeClass}
(3), $E[0,k]\in\mathcal{K}$, and therefore $k\in K$.
\end{proof}

\begin{example}
The class $\mathcal K$ of all EAs is a strong type class of PEAs; hence
by Theorem \ref{th:TypeClass}, the set $K$ of all elements $k\in E$ such
that $E[0,k]$ is an EA is an STD subset of $E$.
\end{example}

\begin{assumption}
From now on we will assume that $K$ is a TD subset of the COGPEA $E$.
\end{assumption}

\begin{definition}
${\widetilde K}:=K\cap\Gamma(E)$.
\end{definition}

\begin{theorem}\label{th:kstar}
There exists $k^*\in K$ such that $\gamma\sb{k^*}$ is the largest
mapping in\linebreak $\{\gamma_k:k\in K\}=\{\gamma_e:e\in E, e
\leq k^*\}=\Theta\sb{\gamma}[0,\gamma\sb{k^*}]$, which is a
sublattice of $\GEX(E)$, and as such, it is a boolean algebra.
Moreover, ${\widetilde K}$ is a TD subset of $E$, there exists
${\widetilde k}\in{\widetilde K}$ such that $\gamma\sb{\widetilde k}$
is the largest mapping in $\{\gamma_k:k\in{\widetilde K}\}=\{\gamma_e:e
\in E,e\leq{\widetilde k}\}=\Theta\sb{\gamma}[0,\gamma\sb{\widetilde k}]$,
which is a sublattice of $\GEX(E)$, and as such, it is a boolean algebra.
\end{theorem}

\begin{proof}
Let us take a maximal $\gamma$-orthogonal family $(k_i)_{i\in I}\subseteq
K$ and set $k^*:=\bigvee_{i\in I} k_i$. Then $k^*\in K$, because $K$ is TD
subset of $E$. Let $k\in K$. As $\gamma\sb{k}k=k$ and $(\gamma\sb{k^*})'
k\leq k$, we have $(\gamma\sb{k}\wedge(\gamma\sb{k^*})')k=\gamma\sb{k}k
\wedge(\gamma\sb{k^*})'k=k\wedge(\gamma\sb{k^*})'k=(\gamma\sb{k^*})'k$.
Also, by Theorem \ref{th:ExCovProp} (vi), $\gamma_k\wedge(\gamma_{k^*})'
=\gamma\sb{d}$, where $d:=(\gamma_{k^*})'k$, and since $K\sp{\gamma}
\subset K$, it follows that ${\widehat k}:=(\gamma\sb{k^*})'k=\gamma\sb{d}k
\in K$ with $\gamma\sb{k^*}{\widehat k}=\gamma\sb{k^*}((\gamma\sb{k^*})'k)=0$.
Therefore, by Theorem \ref{th:ExCovProp} (v), $\gamma\sb{\widehat k}
\wedge\gamma\sb{k\sp{\ast}}=\gamma\sb{\gamma\sb{k\sp{\ast}}{\widehat k}}=0$,
and since $k\sb{i}\leq k\sp{\ast}$, it follows that $\gamma\sb{\widehat k}
\wedge\gamma\sb{k\sb{i}}=0$ for all $i\in I$. Consequently, $(\gamma
\sb{k^*})'k={\widehat k}=0$ by the maximality of $(k_i)_{i\in I}$, therefore
$k=\gamma\sb{k^*}k$, whence $\gamma\sb{k}\leq\gamma\sb{k^*}$.

Suppose $k\in K$ and put $e:=\gamma\sb{k}k^*$. Then $e\leq k^*$ with
$\gamma\sb{k}=\gamma\sb{k}\wedge\gamma\sb{k^*}=\gamma\sb{\gamma\sb{k}k^*}
=\gamma\sb{e}$, whence $\{\gamma\sb{k}:k\in K\}\subseteq\{\gamma\sb{e}
:e\in E, e\leq k^*\}$. If $e\in E$ and $e\leq k^*$, then $\gamma\sb{e}
\leq\gamma\sb{k^*}$, so $\{\gamma\sb{e}:e\in E, e\leq k^*\}\subseteq
\{\gamma\sb{e}:e\in E, \gamma\sb{e}\leq\gamma\sb{k^*}\}=\Theta\sb{\gamma}
[0,\gamma\sb{k^*}]$. Finally, suppose $e\in E$ with $\gamma\sb{e}\leq
\gamma\sb{k^*}$, and put $k:=\gamma\sb{e}k^*$. Since $K$ is TD, we have
$k\in K$; moreover, $\gamma\sb{e}=\gamma\sb{e}\wedge\gamma\sb{k^*}=
\gamma\sb{k}$, so $\Theta\sb{\gamma}[0,\gamma\sb{k^*}]\subseteq
\{\gamma\sb{k}:k\in K\}$.

By Corollary \ref{co:ThetasbgammaGBA}, $\Theta\sb{\gamma}$ is a
generalized boolean algebra; hence the interval $\Theta\sb{\gamma}
[0,\gamma\sb{k^*}]=\{\pi\in\Theta\sb{\gamma}:0\leq\pi\leq\gamma
\sb{k^*}\}$ is a boolean algebra with unit $\gamma\sb{k^*}$.

That ${\widetilde K}$ is a TD subset, follows from Theorem \ref{centrTD}
and the fact that ${\widetilde K}=K\cap\Gamma(E)$. Thus we obtain the second
part of the theorem by applying the first part to ${\widetilde K}$.
\end{proof}

Since $\gamma_{k^*}\in\GEX(E)$ is the largest element in $\{\gamma_k: k
\in K\}$, it is uniquely determined by the TD set $K$. Likewise, ${\widetilde k}$
is uniquely determined by ${\widetilde K}=K\cap\Gamma(E)$, hence it also
is uniquely determined by $K$, and we may formulate the following definition.

\begin{definition} \label{df:gammasbK}
With the notation of Theorem \ref{th:kstar}, (1) $\gamma_K:=\gamma_{k^*}$
and (2) $\gamma_{\widetilde K}:=\gamma_{\widetilde k}$\,.
\end{definition}

\begin{corollary}\label{co:gammasbK}
$\Theta\sb{\gamma}[0,\gamma\sb{K}]$ is a boolean algebra and we have{\rm:}
\begin{enumerate}
\item $\gamma_{\widetilde K}\leq\gamma_K\in\Theta\sb{\gamma}[0,\gamma\sb{K}]
 \subseteq\GEX(E)$.
\item $\gamma_K=\bigvee_{k\in K}\gamma_k$.
\item $\gamma_K$ is the smallest mapping $\pi\in\GEX(E)$
 such that $K\subseteq\pi(E)$.
\item $\gamma_{\widetilde K}=\bigvee_{k\in{\widetilde K}}
 \gamma_k\in\GEX(E)$.
\item $\gamma_{\widetilde K}$ is the smallest mapping $\pi\in
 \Theta_{\gamma}$ such that ${\widetilde K}\subseteq\pi(E)$.
\end{enumerate}
\end{corollary}

\begin{proof}
(i) This is clear by Theorem \ref{th:kstar}, because $\{\gamma_k :k\in
\tilde{K}\}\subseteq\{\gamma_k :k\in K\}$.

(ii) By Theorem \ref{th:kstar}, $\gamma_K$ is the largest mapping in
$\{\gamma\sb{k}:k\in K\}$, from which (ii) follows immediately.

(iii) First we show that $K\subseteq\gamma_K(E)$. Indeed, if $k\in K$,
then $\gamma_k\leq\gamma_K$, so $k=\gamma_k k\leq\gamma_K k\leq k$,
and therefore $k=\gamma_K k\in\gamma_K(E)$. Suppose $K\subseteq\pi(E)$
for some $\pi\in\GEX(E)$. Then, $k^*\in K\subseteq\pi(E)$, so $k^*=
\pi k^*$. But $\gamma_{k^*}$ is the smallest mapping in $\GEX(E)$ with
the latter property, whence $\gamma_{k^*}\leq\pi$.

Proofs of (iv) and (v) are similar to (ii) and (iii) with $\tilde{K}$
instead of $K$.
\end{proof}

\begin{definition} \label{df:Type}
Let $\pi\in\GEX(E)$. Then{\rm:}
\begin{enumerate}
\item[(1)] $\pi$ is \emph{type-$K$} iff there exists $k\in{\widetilde K}$
 such that $\pi=\gamma_k$.
\item[(2)] $\pi$ is \emph{locally type-$K$} iff there exists $k\in K$
 such that $\pi=\gamma_k$.
\item[(3)] $\pi$ is \emph{purely non-$K$} iff $\pi\wedge\gamma_K=0$, i.e.,
 iff $\pi\leq(\gamma\sb{K})'$.
\item[(4)] $\pi$ is \emph{properly non-$K$} iff $\pi\wedge\gamma_
 {\widetilde K}=0$, i.e., iff $\pi\leq(\gamma\sb{\widetilde K})'$.
\end{enumerate}
\end{definition}

\begin{remark}\label{rm:Type}
Directly from Definition \ref{df:Type} and Corollary \ref{co:gammasbK},
we have the following for all $\pi,\xi\in\GEX(E)$:
\begin{enumerate}
\item If $\pi$ is type-$K$, then $\pi$ is locally type-$K$.
\item If $\pi$ is purely non-$K$, then $\pi$ is properly non-$K$.
\item If $\pi$ is both type-$K$ and properly non-$K$, then $\pi=0$.
\item If $\pi$ is both locally type-$K$ and purely non-$K$, then $\pi=0$.
\item If $\xi\in\Theta_{\gamma}$ and $\pi$ is type-$K$ or locally type-$K$
 then so is $\pi\wedge\xi$.
\item If $\pi$ is purely non-$K$ or properly non-$K$, then so is $\pi
 \wedge\xi$.
\item If both $\pi$ and $\xi$ are type $K$, locally type $K$, purely
 non-$K$, or properly non-$K$, then so is $\pi\vee\xi$.
\end{enumerate}
\end{remark}

\begin{theorem} \label{th:Type}
Let $\pi\in\GEX(E)$. Then{\rm:}
\begin{enumerate}
\item $\pi$ is type-$K$ iff $\pi\in\Theta_{\gamma}$ and
 $\pi\leq\gamma_{\widetilde K}$.
\item  If $K$ is STD and $\pi$ is type-$K$, then $\pi(E)\subseteq K$.
\item $\pi$ is locally type-$K$ iff $\pi\in\Theta_{\gamma}$ and $\pi
 \leq\gamma_K$.
\item  If $\pi$ is purely non-$K$, then $K\cap\pi(E)=\{0\}$
\item if $\pi$ is properly non-$K$, then ${\widetilde K}\cap\pi(E)
 =\{0\}$.
\end{enumerate}
\end{theorem}

\begin{proof}
(i) By Theorem \ref{th:kstar} and Definition \ref{df:gammasbK}, $\{\gamma
\sb{k}:k\in{\widetilde K}\}=\Theta\sb{\gamma}[0,\gamma\sb{\widetilde K}]
=\{\gamma\sb{e}:e\in E, \gamma\sb{e}\leq\gamma\sb{\widetilde K}\}$, from
which (i) follows immediately.

(ii) If $\pi$ is type-$K$, then $\pi=\gamma_k$ for some $k\in K\cap\Gamma(E)$,
whence by Theorem \ref{th:gammahullsys}, $\pi=\gamma\sb{k}=\pi\sb{k}$, and
therefore, since $K$ is STD, $\pi(E)=E[0,k]\subseteq K$.

(iv) Suppose that $\pi$ is purely non-$K$, i.e., $\pi\wedge\gamma\sb{K}=0$.
Thus if $k\in K$, then $\gamma\sb{k}\leq\gamma\sb{K}$, whence $\pi\wedge
\gamma\sb{k}=0$. Therefore, if $k\in K\cap\pi(E)$, then $k=k\wedge k=
\pi k\wedge\gamma\sb{k}k=(\pi\wedge\gamma\sb{k})k=0$.

The proofs of (iii) and (v) are analogous to those of (i) and (iv).
\end{proof}

\begin{definition} \label{df:faithful}
An element $f\in E$ is \emph{faithful} iff $\gamma\sb{f}=1$.
\end{definition}

As is easily seen, if $\pi\in\GEX(E)$, then an element $f\in\pi(E)$
is faithful in the GPEA $\pi(E)$ iff $\gamma\sb{f}=\pi$.

\begin{theorem}\label{ksharp}
Let $\pi\in\Theta_{\gamma}$ and put $k^{\sharp}:=\pi k^*$, where
$k^*\in K$ is the element in Theorem \ref{th:kstar}. Then $k^{\sharp}
\in K\cap\pi(E)$ and the following conditions are mutually equivalent{\rm:}
\begin{enumerate}
\item $\pi$ is locally type-$K$.
\item $k^{\sharp}$ is faithful in the direct summand $\pi(E)$ of $E$ {\rm(}i.e.,
 $\gamma_{k^{\sharp}}=\pi${\rm)}.
\item If $\xi\in\Theta_{\gamma}$ with $\xi\wedge\pi\not =0$, then $k^{\sharp}$
 has a nonzero component $0\not=\xi k^{\sharp}$ in the direct summand
 $\xi(\pi(E))$ of the GPEA $\pi(E)$, and $\xi k^{\sharp}\in K$.
\end{enumerate}
\end{theorem}

\begin{proof}
As $\pi\in\Theta\sb{\gamma}$, there exists $d\in E$ with $\pi=\gamma_d$. Since
$K$ is TD and $k^*\in K$, we have $k\sp{\sharp}=\pi k^*=\gamma\sb{d}k^*\in K$.
Also, $k\sp{\sharp}=\pi k^*\in\pi(E)$, whence $k^{\sharp}\in K\cap\pi(E)$.

(i) $\Rightarrow$ (ii): If $\pi=\gamma_d$ is locally type-$K$, then $\gamma_d
\leq\gamma_K=\gamma_{k^*}$ so $\gamma_{k^{\sharp}}=\gamma_{\gamma_d k^*}=
\gamma_d\wedge\gamma_{k^*}=\gamma\sb{d}=\pi$.

(ii) $\Rightarrow$ (iii): Assume (ii) and the hypotheses of (iii). Then
$\xi k^{\sharp}=\xi\pi k^*\in\xi(\pi(E))$, $\gamma\sb{k\sp{\sharp}}=\pi$,
there exists $e\in E$ with $\xi=\gamma\sb{e}$, and $0\not=\xi\wedge\pi=
\gamma\sb{e}\wedge\gamma\sb{k\sp{\sharp}}=\gamma\sb{\gamma\sb{e}k\sp{\#}}
=\gamma\sb{\xi k\sp{\#}}$, so $\xi k\sp{\#}\not=0$. Also, since $K$ is TD
and $k\sp{\#}\in K$, we have $\xi k\sp{\#}=\gamma\sb{e}k\sp{\#}\in K$.

(iii) $\Rightarrow$ (i): Assume (iii). We have $\pi=\gamma\sb{d}$, and
since $k\sp{\#}\in K$, we also have $\gamma\sb{k\sp{\#}}\leq\gamma\sb{K}$;
hence, by Theorem \ref{th:Type} (iii), it will be sufficient to show that
$\gamma\sb{d}\leq\gamma\sb{k\sp{\#}}$. Aiming for a contradiction, we
assume that $\gamma\sb{d}\not\leq\gamma\sb{k\sp{\#}}$, i.e., by Theorem
\ref{th:ExCovProp} (vi),  $\xi:=\gamma\sb{e}=(\gamma_{k^{\sharp}})'\wedge
\gamma_d\not=0$, where $e:=\gamma\sb{k^{\#}}d$. Then $\xi\leq\gamma\sb{d}
=\pi$, so $\xi\wedge\pi=\xi\not=0$. But $\xi\leq(\gamma_{k^{\sharp}})'$
implies $\xi k^{\sharp}=0$, contradicting (iii).
\end{proof}

\begin{corollary}
If $\pi\in\GEX(E)$ is locally type-$K$ and $\xi\in\Theta_{\gamma}$ with
$\xi\wedge\pi\not =0$, then the direct summand $\xi(\pi(E))$ of $\pi(E)$
contains a nonzero element of $K$.
\end{corollary}

\begin{proof}
The nonzero element $\xi k^{\sharp}\in K$ in Theorem \ref{ksharp} belongs to
$\xi(\pi(E))$ .
\end{proof}

\begin{lemma}\label{le:type}
\rm{(i)} There exists a unique mapping $\pi\in\Theta_{\gamma}$, namely
$\pi=\gamma_K$, such that $\pi$ is locally type-$K$ and $\pi\,'$ is purely
non-$K$. \rm{(ii)} There exists a unique mapping $\xi\in\Theta_{\gamma}$,
namely $\xi=\gamma_{\widetilde K}$, such that $\xi$ is type-$K$ and $\xi\,'$
is properly non-$K$.
\end{lemma}

\begin{proof}
By Theorem \ref{th:Type} (iii), $\pi$ is locally type-$K$ iff $\pi
\leq\gamma_K$ and by Definition \ref{df:Type} (3), $\pi\,'$ is purely
non-$K$ iff $\pi\,'\wedge\gamma_K=0$, i.e., iff $\gamma_K\leq\pi$, from
which (i) follows. Similarly, (ii) follows from Theorem \ref{th:Type} (i)
and Definition \ref{df:Type} (4).
\end{proof}

\section{Type-decomposition of COGPEA} \label{sc:TypeDecomp} 

\noindent\emph{ We maintain our standing hypothesis that $K$ is a TD subset of
the COGPEA $E$.} According to Lemma \ref{le:type}, we have two bipartite direct
decompositions $E=\pi(E)\oplus\pi\,'(E)$ and $E=\xi(E)\oplus\xi\,'(E)$,
corresponding to $\pi=\gamma_K$ and $\xi=\gamma_{\widetilde K}$. Thus we
may decompose $E$ into four direct summands:
\[
E=(\pi\wedge\xi)(E)\oplus(\pi\wedge\xi\,')(E)\oplus(\pi\,'\wedge\xi)(E)
 \oplus(\pi\,'\wedge\xi\,')(E)
\]
one of which, namely $(\pi\,'\wedge\xi)(E)$ is necessarily $\{0\}$,
because by Corollary \ref{co:gammasbK} (i), $\xi\leq\pi$. Therefore we
have the following \emph{fundamental direct decomposition theorem for
a COGPEA $E$ with a TD set $K\subseteq E$}.

\begin{theorem}\label{th:decompos}
There exist unique pairwise disjoint mappings $\pi_1, \pi_2, \pi_3\in\GEX(E)$,
namely $\pi_1=\gamma_{\widetilde K}, \pi_2=\gamma_K\wedge(\gamma_
{\widetilde K})'$, and $\pi_3=(\gamma_K)'$, such that{\rm}:
\begin{enumerate}
\item $\pi_1\vee\pi_2\vee\pi_3=1$ so that $E=\pi_1(E)\oplus\pi_2(E)\oplus
 \pi_3(E)$, and
\item $\pi_1$ is type-$K$, $\pi_2$ is locally type-$K$ but properly non-$K$,
 and $\pi_3$ is purely non-$K$.
\end{enumerate}
\end{theorem}

\begin{proof}
For the existence part of the theorem, put $\pi_1=\gamma_{\widetilde K},
\pi_2=\gamma_K\wedge(\gamma_{\widetilde K})'$, and $\pi_3=(\gamma_K)'$.
Obviously, $\pi\sb{1}$, $\pi\sb{2}$, and $\pi\sb{3}$ are pairwise
disjoint, and since $\gamma\sb{\widetilde K}\leq\gamma\sb{K}$, it is
clear that $\pi\sb{1}\vee\pi\sb{2}\vee\pi\sb{3}=1$. Evidently, $\pi
\sb{1}\in\Theta\sb{\gamma}$, and by Theorem \ref{th:ExCovProp} (viii),
$\pi\sb{2}\in\Theta\sb{\gamma}$. Thus, by Theorem \ref{th:Type} (i),
$\pi\sb{1}$ is type $K$, and by Theorem \ref{th:Type} (iii) $\pi\sb{2}$
is locally type $K$. Also, by parts (3) and (4) of Definition \ref{df:Type},
$\pi\sb{3}$ is purely non-$K$ and $\pi\sb{2}$ is properly non-$K$.

To prove uniqueness, suppose that $\pi_{1a},\pi_{2a},\pi_{3a}$ are pairwise
disjoint mappings in the boolean algebra $\GEX(E)$ satisfying (i) and (ii).
Then $\pi_{1a}\leq\gamma_{\tilde{K}}$ by Theorem \ref {th:Type} (i),
$\pi_{2a}\leq\gamma_K\wedge(\gamma_{\tilde{K}})'$ by Theorem \ref{th:Type}
(iii) and Definition \ref{df:Type} (4), and $\pi_{3a}\leq(\gamma_K)'$ by
Definition \ref{df:Type} (3). Thus after an elementary boolean computation,
we finally get $\pi_1=\pi_{1a}$, $\pi_2=\pi_{2a}$ and $\pi_3=\pi_{3a}$.
\end{proof}

In what follows we will obtain a decomposition of the COGPEA $E$ into types
I, II and III analogous to the type decomposition of a von Neumann algebra.
We shall be dealing with two TD subsets $K$ and $F$ of $E$ such that $K\subseteq F$.
For the case in which $E$ is the projection lattice of a von Neumann algebra,
one takes $K$ to be the set of abelian elements and $F$ to be the set of
finite elements in $E$.

Thus, in what follows, assume that $K$ and $F$ are TD subsets of the COGPEA
$E$ such that $K\subseteq F$. By Theorem \ref{th:decompos}, we decompose $E$ as

$E=\pi_1(E)\oplus \pi_2(e)\oplus\pi_3(E)$ and also as $E=\xi_1(E)\oplus
 \xi_2(E)\oplus\xi_3(E)$ where
\begin{eqnarray*}
&\pi_1=\gamma_{\widetilde K},\  \pi_2=\gamma_K\wedge(\gamma_{\widetilde K})',\
 \pi_3=(\gamma_K)',\\
&\xi_1=\gamma_{\widetilde F},\  \xi_2=\gamma_F\wedge(\gamma_{\widetilde F})',\
 \xi_3=(\gamma_F)'.
\end{eqnarray*}
As $K\subseteq F$, it is clear that $\gamma_K\leq\gamma_F$, $\gamma_
{\widetilde K}\leq\gamma_{\widetilde F}$, $(\gamma_F)'\leq(\gamma_K)'$,
and $(\gamma_{\widetilde F})'\leq(\gamma_{\widetilde K})'$.

Applying Theorem \ref{th:decompos}, we obtain a direct sum decomposition
\[
E=\tau_{11}(E)\oplus\tau_{21}(E)\oplus\tau_{22}(E)\oplus\tau_{31}(E)
 \oplus\tau_{32}(E)\oplus\tau_{33}(E),
\]
where $\tau_{ij}=\pi_i\wedge \xi_j$, for $i,j=1,2,3$. Evidently, $\tau_{11}
=\pi_1$, $\tau_{33}=\xi_3$ and $\tau_{12}=\tau_{13}=\tau_{23}=0$.

\begin{definition}\label{de:I,II,III} (\cite[Definition 6.3]{HandD},
\cite[Definition 13.3]{HDandTD})
Let $\pi\in\GEX(E)$. For the TD sets $K$ and $F$ with $K\subseteq F$:
\begin{itemize}
\item  $\pi$ is \emph{type-I} iff it is locally type-$K$, i.e., iff
 $\pi \in \Theta_{\gamma}$ and $\pi \leq\gamma_K$.
\item $\pi$ is \emph{type-II} iff it is locally type-$F$, but purely
 non-$K$, i.e., iff $\pi\in \Theta_{\gamma}$ and $\pi \leq \gamma_F
 \wedge(\gamma_K)'$.
\item $\pi$ is \emph{type-III} if it is purely non-$F$, i.e., iff
 $\pi \leq (\gamma_F)'$.
\item $\pi$ is \emph{type-I$_F$} (respectively, \emph{type-II$_F$})
 iff it is type-I (respectively, type-II) and also type-$F$, i.e.,
 iff $\pi\in\Theta_{\gamma}$ and $\pi\leq\gamma_K\wedge\gamma_{\tilde{F}}$
 (respectively, $\pi\in\Theta_{\gamma}$ and $\pi\leq\gamma_F\wedge
 (\gamma_K)'\wedge\gamma_{\tilde{F}}$).
\item $\pi$ is \emph{type-I$_{\neg F}$} (respectively, \emph{type-II$_
 {\neg F}$}) iff it is type-I (respectively, type-II) and also properly
 non-$F$, i.e., iff $\pi\in\Theta_{\gamma}$ and $\pi\leq\gamma_K\wedge
 (\gamma_{\tilde{F}})'$ (respectively, iff $\pi \in \Theta_{\gamma}$ and
 $\pi\leq \gamma_F\wedge (\gamma_K)'\wedge (\gamma_{\tilde{F}})'$).
\end{itemize}
If $\pi$ is type-I, type-II, etc. we also say that the direct summand
$\pi(E)$ is type-I, type-II, etc.
\end{definition}

The following theorem is the I/II/III - decomposition theorem for COGPEAs.

\begin{theorem}\label{th:I-II-III} Let $E$ be COGPEA and let $K$ and $F$ be
TD sets in $E$ with $K\subseteq F$. Then there are pairwise disjoint mappings
$\pi_I, \pi_{II}, \pi_{III}\in \GEX(E)$ of types I, II and III, respectively,
such that $E$ decomposes as a direct sum
$$
E=\pi_I(E)\oplus \pi_{II}(E)\oplus \pi_{III}(E).
$$
Such a direct sum decomposition is unique and
$$
\pi_I=\gamma_K,\ \pi_{II}=\gamma_F\wedge(\gamma_K)',\  \pi_{III}=(\gamma_F)'.
$$
Moreover, there are further decompositions
$$
\pi_I(E)=\pi_{I_F}(E)\oplus \pi_{I_{\neg F}}(E), \ \pi_{II}(E)=
 \pi_{II_F}(E)\oplus \pi_{II_{\neg F}}(E),
$$
where $\pi_{I_F}, \pi_{I_{\neg F}}, \pi_{II_F},\pi_{II_{\neg F}}$ are of
types $I_F, I_{\neg F}, II_F, II_{\neg F}$, respectively. These decompositions
are also unique and
\[
\pi_{I_F}=\gamma_K\wedge \gamma_{\widetilde F},\  \pi_{I_{\neg F}}=
 \gamma_K\wedge (\gamma_{\widetilde F})',
\]
\[
\hspace{1.6 cm}\pi_{II_F}=\gamma_{\widetilde F}\wedge (\gamma_K)',\
 \pi_{II_{\neg F}}=\gamma_F\wedge(\gamma_{\widetilde F})'\wedge (\gamma_K)'.
\]
\end{theorem}

\begin{proof} For the existence part of the theorem, we put $\pi_I:=
\tau_{11}\vee \tau_{21}\vee \tau_{22}$, $\pi_{II}:=\tau_{31}\vee
\tau_{32}$, $\pi_{III}:=\tau_{33}$, $\pi_{I_F}:=\tau_{11}\vee\tau_{21}$,
$\pi_{I_{\neg F}}:=\tau_{22}$, $\pi_{II_F}:=\tau_{31}$, and $\pi_
{II_{\neg F}}:=\tau_{32}$. Evidently, all the required conditions are
satisfied. The proof of uniqueness is also straightforward.
\end{proof}

\end{document}